\documentclass[a4paper,11pt]{article}

\usepackage{a4wide}
\usepackage{hyperref}
\hypersetup{colorlinks=true,citecolor=blue,linkcolor=blue,urlcolor=blue}
\usepackage{amsmath}
\usepackage{amssymb}
\usepackage{amsthm}   
\usepackage{enumerate}
\usepackage{thm-restate}
\usepackage{tikz}
\usepackage{tikz-cd}
\usepackage{svg}

\theoremstyle{plain}
\newtheorem{theorem}{Theorem}
\newtheorem*{theorem*}{Theorem}
\newtheorem{corollary}[theorem]{Corollary}
\newtheorem{lemma}[theorem]{Lemma}
\newtheorem{observation}[theorem]{Observation}
\newtheorem{prop}[theorem]{Proposition}

\theoremstyle{definition}
\newtheorem{definition}[theorem]{Definition}

\usepackage[UKenglish]{babel}
\usepackage[UKenglish]{isodate}
\usepackage{url}
\usepackage{soul}

\let\epsilon=\varepsilon
\let\phi=\varphi

\newcommand{\Hom}{\ensuremath{\mathrm{Hom}}}
\newcommand{\Ans}{\ensuremath{\mathrm{Ans}}}
\newcommand{\Aut}{\ensuremath{\mathrm{Aut}}}

\newcommand{\Ctw}{\ensuremath{\mathsf{ew}}}
\newcommand{\Tw}{\ensuremath{\mathsf{tw}}}
\newcommand{\tw}{\Tw}
\newcommand{\ctw}{\Ctw}
\newcommand{\etw}{\Ctw}
\newcommand{\sew}{\ensuremath{\mathsf{sew}}}
\newcommand{\hsew}{\ensuremath{\mathsf{hsew}}}

\newcommand{\calB}{\ensuremath{\mathcal{B}}}
\newcommand{\calT}{\ensuremath{\mathcal{T}}}
\newcommand{\calN}{\ensuremath{\mathcal{N}}}

\newcommand{\Ext}{\ensuremath{\mathrm{Ext}}}

\def\newgraph{\mathcal{G}}
\def\newcol{\mathcal{C}}  
\def\vv{\Vec{v}}
\def\vz{\Vec{z}}
\def\calE{\mathcal{E}}
\def\MD{M_{Y}}

\def\M12{M_{X\mathrm{D}}}
\def\hatM12{\hat{M}_{X\mathrm{D}}}

\newcommand{\cpHom}{\mathsf{cpHom}}
\newcommand{\cpAns}{\mathsf{cpAns}}

\def\bij{\mathsf{Bij}}

\date{11 March 2024}

\title{\vspace{-2.5cm}The Weisfeiler-Leman Dimension of Conjunctive Queries\footnote{Andreas Göbel was funded by the project PAGES (project No. 467516565) of the German Research Foundation (DFG). For the purpose of Open Access, the authors have applied a CC BY public copyright licence to any Author Accepted Manuscript version arising from this submission. All data is provided in full in the results section of this paper.}}

\author{
Andreas G\"obel \\ Hasso Plattner Institute \\ University of Potsdam\\ Germany \and
Leslie Ann Goldberg \\ Department of Computer Science\\ University of Oxford\\ United Kingdom
\and
Marc Roth \\ School of Electronic Engineering and Computer Science \\ Queen Mary University of London\\ United Kingdom 
}
\begin{document}
\maketitle

\begin{abstract}
A graph parameter is a function~$f$ on graphs 
with the property that, for any pair of isomorphic graphs~$G_1$ and~$G_2$, $f(G_1)=f(G_2)$.
The Weisfeiler--Leman (WL) dimension of~$f$ is the minimum $k$ such that, 
if $G_1$ and $G_2$ are   indistinguishable by the $k$-dimensional WL-algorithm 
then $f(G_1)=f(G_2)$. The
WL-dimension of~$f$ is  $\infty$ if no such $k$ exists. We study the WL-dimension of 
graph parameters characterised by the number of 
answers from a fixed conjunctive query to the graph. Given a conjunctive query $\varphi$, 
we quantify
the WL-dimension of the function that maps every graph~$G$ to the number of answers of $\varphi$ in $G$.
    
The works of Dvor{\'{a}}k (J.\ Graph Theory 2010), Dell, Grohe, and Rattan (ICALP 2018), and Neuen (ArXiv 2023) have answered this question for \emph{full} conjunctive queries, which are conjunctive queries without existentially quantified variables. For such queries $\varphi$, the WL-dimension is equal to the treewidth of the Gaifman graph of $\varphi$.

In this work, we give a characterisation that applies to all conjunctive qureies. Given any conjunctive query $\varphi$, we prove that its WL-dimension is equal to the \emph{semantic extension width} $\sew(\varphi)$, a novel width measure that can be thought of as a combination of the treewidth of $\varphi$ and its \emph{quantified star size}, an invariant introduced by Durand and Mengel (ICDT 2013) describing how the existentially quantified variables of $\varphi$ are connected with the free variables. Using the recently established equivalence between the WL-algorithm and higher-order Graph Neural Networks (GNNs) due to Morris et al.{} (AAAI 2019), we obtain as a consequence that the function counting answers to a conjunctive query $\varphi$ cannot be computed by GNNs of order smaller than $\sew(\varphi)$.

The majority of the paper is concerned with establishing a lower bound of the WL-dimension of a query. Given any conjunctive query $\varphi$ with semantic extension width $k$, we consider a graph $F$ of treewidth $k$ obtained from the Gaifman graph of $\varphi$ by repeatedly cloning the vertices corresponding to existentially quantified variables. Using a recent modification 
due to Roberson (ArXiv 2022) 
of the Cai-Fürer-Immerman construction (Combinatorica 1992), we then obtain a pair of graphs $\chi(F)$ and $\hat{\chi}(F)$ that are indistinguishable by the $(k-1)$-dimensional WL-algorithm since $F$ has treewidth $k$. Finally, in the technical heart of the paper, we show that $\varphi$ has a different number of answers in $\chi(F)$ and $\hat{\chi}(F)$. Thus, $\phi$ can distinguish two graphs that cannot be distinguished by the $(k-1)$-dimensional WL-algorithm, so the WL-dimension of $\phi$ is at least~$k$.
 
\end{abstract}

\section{Introduction}
The Weisfeiler--Leman (WL) algorithm~\cite{WL68} and its higher dimensional generalisations~\cite{CaiFI92} are amongst the most well-studied heuristics for graph isomorphism. 
This algorithm works as follows.
For each positive integer $k$, the $k$-dimensional WL-algorithm iteratively maps $k$-tuples of vertices of a graph to multisets of colours. 
Two graphs $G$ and $G'$ are said to be $k$-WL-equivalent, denoted $G\cong_k G'$, if this algorithm returns the same vertex colouring 
for $G$ and $G'$, up to consistently renaming the colours. For the specific case $k=1$ the WL-algorithm is equivalent to the colour-refinement algorithm, which  is  a widely used and efficiently-implementable heuristic for graph isomorphism (see e.g.\ \cite{Arvind16,GroheKMS21}).

In addition to applications to graph isomorphism, recent works have shown that the expressiveness of Graph Neural Networks (GNNs) and their higher order generalisations is precisely characterised by the WL-algorithm~\cite{Morrisetal19,XuHLJ19}. This result has sparked a flurry of research with the objective of determining which graph parameters  are invariant on graphs that are indistinguishable by the WL-algorithm~\cite{MorrisRM20,ChenCVB20,BarceloGRR21,ArvindFKV22,neuen2023homomorphism,LanzingerB23,BouritsasFZB23}. We refer the reader to the survey by Grohe~\cite{Grohe21} for further reading.

Over the years, surprising alternative characterisations of $k$-WL-equivalence have been established.

\begin{itemize}
\item[(I)] $G\cong_1 G'$ if and only if $G$ and $G'$ are fractionally isomorphic~\cite{Tinhofer86,Tinhofer91}.
\item[(II)] For each positive integer $k$, $G \cong_k G'$ if and only if there is no first-order formula with counting quantifiers that uses at most $k+1$ variables and that can distinguish~$G$ and~$G'$~\cite{ImmermanL90,CaiFI92}.
\item[(III)] For each positive integer $k$, $G \cong_k G'$ if and only if, for each graph $H$ of treewidth at most~$k$, the number of graph homomorphisms from $H$ to $G$ 
is equal to the number of graph homomorphisms from $H$ to $G'$~\cite{Dvorak10,DellGR18}. 
This is the characterisation of $k$-WL-equivalence
that will be used in this work (see Definition~\ref{def:wl_equivalence}).
\end{itemize}

The characterisation in (III) has ignited  interest in studying the \emph{WL-dimension} of counting graph homomorphisms and 
of counting related patterns~\cite{ChenCVB20,ArvindFKV22,neuen2023homomorphism,LanzingerB23,BouritsasFZB23}.

A graph parameter $f$ is a function from
graphs 
 that is invariant under isomorphisms.  
The WL-dimension of 
a graph parameter
$f$ is the minimum positive integer $k$ such that $f$ cannot distinguish $k$-WL-equivalent graphs (see Definition~\ref{def:WLdim}).
Building upon the works of Dvor{\'{a}}k~\cite{Dvorak10}, Dell, Grohe and Rattan~\cite{DellGR18}, Roberson~\cite{roberson2022oddomorphisms}, and Seppelt~\cite{Seppelt23}, it has very recently been shown by Neuen~\cite{neuen2023homomorphism} that the WL-dimension of
the graph parameter that
counts homomorphisms from a fixed graph $H$ is exactly the treewidth of $H$. It is well-known that counting homomorphisms is equivalent to counting answers to conjunctive queries without existentially quantified variables (see e.g.\ \cite{PichlerS13}); such conjunctive queries are also called \emph{full} conjunctive queries.
In this work, we 
consider all conjunctive queries including those that have 
existentially quantified variables and we answer the fundamental question:  {\it What is the WL-dimension of 
the graph parameter that
counts answers to fixed conjunctive queries?}

\noindent To state our results, we first introduce some central concepts.

\subsection{Conjunctive Queries and Semantic Extension Width}
A \emph{conjunctive query} $\varphi$ consists of a set of free variables $X=\{x_1,\dots,x_k\}$ and a set of (existentially) quantified variables $Y=\{y_1,\dots,y_\ell\}$, and is of the form
\[ \varphi(x_1,\dots,x_k) = \exists y_1,\dots, y_\ell \colon A_1 \wedge \dots \wedge A_m,\] 
such that each $A_i$ is an atom $R(\Vec{z})$ where $R$ is a relation symbol and $\Vec{z}$ is a vector of variables in $X \cup Y$. Since we focus in this work on undirected graphs without self-loops, 
in our setting there is only one binary relation symbol $E$, so all atoms are of the form $E(z_1,z_2)$. An \emph{answer} to $\varphi$ in a graph $G$ is an assignment $a$ from the free variables $X$ to $V(G)$ such 
there is an assignment $h$ from $X\cup Y$ to $V(G)$ which agrees with $a$ on~$X$ and has the property that,  for each atom $E(z_1,z_2)$, the image $\{h(z_1),h(z_2)\}$ is an edge of $G$.

As is common in the literature (see e.g.\ \cite{PichlerS13,ChenM15,ChenM16,DellRW19}), we can equivalently express the answers of $\varphi$ in a graph $G$ as partial homomorphisms to $G$. Let $H$ be the graph with vertex set $X \cup Y$  that has as edges the pairs of variables in $X \cup Y$ that occur in a common atom. Then the answers of $\varphi$ in $G$ are the mappings $a\colon X \to V(G)$ that can be extended to a homomorphism from $H$ to $G$. For this reason, following the notation of~\cite{DellRW19}, we will from now an refer to a conjunctive query as a pair $(H,X)$ where $H$ is a graph and $X$ is a subset of vertices of $H$ corresponding to the free variables. 
We will say that $(H,X)$ is \emph{connected} if $H$ is a connected graph.
We will write $\Ans((H,X),G)$ for the set of answers of $(H,X)$ in $G$; this is made formal in Section~\ref{sec:cqs}. 
The WL-dimension of a conjunctive query $(H,X)$ is the WL-dimension of the graph parameter that maps every graph $G$
to   $|\Ans((H,X),G)|$.

\paragraph*{Semantic Extension Width}
Let $(H,X)$ be a conjunctive query and let $Y=V(H)\setminus X$. The graph $\Gamma(H,X)$  is obtained from $H$ by adding an edge between  each  pair of vertices $u\neq v$ in $X$ if and only if there is a connected component in $H[Y]$ that is adjacent to both $u$ and $v$. We then define the \emph{extension width} of $(H,X)$ as the treewidth of $\Gamma(H,X)$; the definition of treewidth can be found in Section~\ref{sec:width}.

The \emph{semantic extension width} of a conjunctive query $(H,X)$, denoted by $\sew(H,X)$ is then the minimum extension width of any conjunctive query $(H',X')$ that is \emph{counting equivalent} to $(H,X)$, i.e., 
any conjunctive query $(H',X')$ such that, for every graph $G$,
$|\Ans((H,X),G)|=|\Ans((H',X'),G)|$. A discussion of counting equivalence and \emph{counting minimal} conjunctive queries can be found in Section~\ref{sec:cqs}.

Before stating our main result, we provide an example of a conjunctive query and its semantic extension width: Let $(S_k,X_k)$ be the $k$-\emph{star query}: $X_k=\{x_1,\dots,x_k\}$ and $S_k$ has vertices $X_k \cup \{y\}$ and edges $\{x_i,y\}$ for all $i\in[k]$. Note that the answers of $(S_k,X_k)$ in a graph $G$ are precisely the assignments from $X_k$ to $V(G)$ 
such that the vertices all of the images of vertices in~$X_k$ have a common neighbour. The $k$-star query is acyclic (i.e., $S_k$ has treewidth $1$) and it has played an important role as a base case for complexity classifications regarding counting answers to conjunctive queries~\cite{ChenM15,DellRW19}. The graph $\Gamma(S_k,X_k)$ is the $(k+1)$-clique which has treewidth $k$. Since it is also minimal with respect to counting equivalence, we have $\sew(S_k,X_k)=k$.

\subsection{Our Contributions}
We  now state our main result.
\begin{restatable}{theorem}{mthm}\label{thm:main_thm}
Let $(H,X)$ be a connected conjunctive query with $X\neq \emptyset$. Then the WL-dimension of $(H,X)$ is equal to its semantic extension width $\sew(H,X)$.
 \end{restatable}

In Theorem~\ref{thm:main_thm}, the WL-dimension of $(H,x)$ is captured by its semantic extension width rather than by its extension width, which is the treewidth of $\Gamma(H,X)$. This is because $H[Y]$ may contain a high-treewidth subgraph that does not influence the number of answers.

As an immediate corollary of Theorem~\ref{thm:main_thm}, we obtain the following alternative characterisation of WL-equivalence.
\begin{corollary}\label{cor:into_WL_char}
For each positive integer $k$, two graphs $G$ and $G'$ are $k$-WL-equivalent if and only if, for each connected conjunctive query $(H,X)$ with $X \neq \emptyset$ and $\sew(H,X)\leq k$,  $|\Ans((H,X),G)|=|\Ans((H,X),G')|$.
\end{corollary}

As the following sections show, our classification of the WL-dimension of conjunctive queries  has further strong consequences  regarding the expressive power of graph neural networks (GNNs), the parameterised complexity of counting answers to conjunctive queries, and  the WL-dimension of first-order formulas with universal quantifiers such as 
the formula corresponding to  dominating sets.

\paragraph*{GNNs and Conjunctive Queries}

 Over  the last decade, GNNs have received increasing attention due to their application to computations involving  graph structured data (see~\cite{LanzingerB23}). Motivated by the fact that the number of occurrences of small patterns can capture interesting global information about graphs, and can therefore be used to compare graphs~\cite{Miloetal02,Alonetal08,JinKL18}, 
researchers have studied the extent to which   GNNs (and their higher order generalisations~\cite{Morrisetal19}) are able to count selected small patterns such as homomorphisms~\cite{LanzingerB23}, subgraphs~\cite{BouritsasFZB23}, and induced subgraphs~\cite{ChenCVB20}.  

Following \cite{Morrisetal19} 
but simplifying the notation for our needs, 
we
represent a $t$-layer order-$k$ GNN $N$  as a 
tuple $N=(G, W_1,\ldots,W_t, f_0, \ldots,f_t)$ where $G$ is a graph, 
each $W_i$ is a set of weights, and
each $f_i$ assigns a feature vector to each $k$-tuple of nodes.
The GNN specifies how $f_i$ is computed from
$G,W_1,\ldots,W_{i-1},f_0,\ldots,f_{i-1}$.
We use $f_N(G)$ to denote the final feature vector so $f_N(G)=f_t$.
The feature vector $f_N(G)$ induces a partition on the $k$-tuples of vertices of $G$, which we call $P_N(G)$.

We next explain what we mean when we say that a GNN can ``compute'' a function on graphs. So far, this has been studied in a somewhat limited context.
Namely, we say that a GNN can ``compute'' a function $A_N(G)$
if $A_N(G)$ can be computed in polynomial time from  $P_N(G)$.
Thus, when we say that a GNN can count small patterns, we mean that the number of such patterns can be efficiently computed  from $P_N(G)$.   We do not address the issue of whether the GNN can itself do the polynomial-computation that is needed to compute $A_N(G)$ from $P_N(G)$. 
Issues of dimension are also beyond the scope of this paper --- in our setting the feature vector induces a partition on the $k$-tuples of vertices of $G$ --- for a brief discussion about how the dimension can be reduced see   \cite{Morrisetal19}.

We say that a GNN $N=(G, W_1,\ldots,W_t, f_0, \ldots,f_t)$ is ``fully refined'' if there is no GNN $N'=(G, W'_1,\ldots,W'_t, f_0,f'_1, \ldots,f'_{t'})$ 
such that $P_{N'}(G)$ strictly refines $P_N(G)$.

In this setting, Morris et al.\ \cite{Morrisetal19} established an equivalence between the expressive power of fully-refined order-$k$ GNNs and the $k$-dimensional WL algorithm.  
For this, let $\calN_k$ be the set of fully-refined order-$k$ GNNs. 
Propositions~3 and~4 of~\cite{Morrisetal19} 
give the following proposition.
\begin{prop}\label{prop:morris}
For all  $N\in \calN_k$,  $P_N(G)$  is exactly the the same as    the partition  
$P_{\text{WL}}(G)$ on $k$-tuples of vertices that is
computed by $k$-WL when it is run with input $G$ and the initial partition induced by the initial feature vector $f_0$ of $N$.
\end{prop}

 Building upon  Proposition~\ref{prop:morris}, the works of Dvor{\'{a}}k~\cite{Dvorak10}, Dell, Grohe and Rattan~\cite{DellGR18}, and Lanzinger and Barcelo~\cite{LanzingerB23} determine the expressiveness of fully refined GNNs in the context of  homomorphism counting. Essentially, order-$k$ GNNs can count homomorphisms from a graph $H$ if and only if the treewidth of $H$ is at most $k$. The ``if'' direction has already been used implicitly in~\cite{Dvorak10,DellGR18}. It follows explicitly from \cite[Theorem 6 and Lemma 7]{LanzingerB23}. The ``only if'' direction follows by combining 
 Proposition~\ref{prop:morris}  with the upper and lower bounds on the WL dimension of counting homomorphisms~\cite{Dvorak10,DellGR18,roberson2022oddomorphisms,LanzingerB23}.
 Specifically, Lanzinger and Barcelo~\cite{LanzingerB23}
show that homomorphisms from $H$ to $G$ can be efficiently computed from the vertex refinement produced when WL-$k$ is run on input~$G$ starting from the partition in which each $k$-tuple 
is assigned a part based on the subgraph that induces.

Our classification (Theorem~\ref{thm:main_thm}) provides  a similar picture in the context of counting answers to conjunctive queries.
First, we will show  that if $(H,X)$ is a conjunctive query with
 $\sew(H,X)=k$ then for all graphs $G$ there is a fully refined GNN $N\in \calN_k$ with underlying graph $G$  that computes $ |\Ans((H,X),G)|$.
This follows from the following two observations.
\begin{enumerate}
\item From \cite[Theorem 6]{LanzingerB23} and Proposition~\ref{prop:morris}, for all $k$, all sequences $F_1,\ldots,F_n$ of graphs of treewidth at most~$k$, all sequences $\mu_1,\ldots,\mu_k$ 
of rational numbers, and all graphs~$G$ 
there is a fully refined GNN $N\in \calN_k$ with underlying graph $G$ 
such that $\sum_{i=1}^n \mu_i |\Hom(F_i,G)|$ can be efficiently computed from $P_N(G)$.

\item   From our work (see Observation~\ref{obs_for_intro}), for all graphs $G$ there is a finite sequence of graphs $F_1,\ldots,F_n$ of treewidth at most~$k$, such that
$|\Ans((H,X),G)|$ can be written as such as sum.

\end{enumerate}

For the other direction we will show that if a fully refined GNN can compute the number of answers from $(H,X)$ then the order of this GNN is at least $\sew(H,X)$.
The proof is based on the following two observations.
 
\begin{enumerate}
\item[(1)] From Proposition~\ref{prop:morris},
 for  all graphs $G'$ and $G''$ such that $G'\cong_k G''$ and all GNNs $N',N''\in \calN_k$ with underlying graphs $G'$ and $G''$, and any function $A_N(G)$ that is efficiently computable  from $P_N(G)$,
$A_{N'}(G')=A_{N''}(G'')$.
\item[(2)] Let $(H,X)$ be a conjunctive query with $\sew(H,X)=k$.
From our Theorem~\ref{thm:main_thm}, there are graphs $G$ and $G'$ such that $G \cong_{k-1} G'$ and $|\Ans((H,X),G)|\neq |\Ans((H,X),G')|$.
\end{enumerate}

We can use these two facts to show that if a fully refined GNN can compute the number of answers from $(H,X)$ then its order is at least $\sew(H,X)$.
In particular, consider $(H,X)$ with $\sew(H,X)=k$.
Suppose for contradiction that, for some $j<k$, some GNN $N \in \calN_j$
can compute $A_N(G) = |\Ans((H,X),G)|$. 
For all $G$ and $G'$ with $G\cong_{k-1} G'$ we have $G\cong_j G'$ so from (1), we have $|\Ans((H,X),G)|= |\Ans((H,X),G')|$, contradicting (2).

\paragraph*{Parameterised counting of answers to conjunctive queries}

The next consequence of our main result reveals a surprising connection between the complexity of counting answers to conjunctive queries and their WL-dimension.
Given a class of conjunctive queries $\Psi$, the counting problem $\#\textsc{CQ}(\Psi)$ takes as input a pair consisting of a conjunctive query $(H,X)\in \Psi$ and a graph $G$ and outputs $|\Ans((H,X),G)|$. 
We say that a class of conjunctive queries has \emph{bounded} WL-dimension if there is a constant $B$ that upper bounds the WL-dimension of all queries in the class. 
The assumption $\mathrm{FPT} \neq W[1]$   is the central (and widely accepted) hardness assumption in parameterised complexity theory (see e.g.\ \cite{FlumG06}).
We say that a conjunctive query is \emph{counting minimal} if it is a minimal representative with respect to counting equivalence (see Definition~\ref{def:zzz}).
Theorem~\ref{thm:main_thm} implies Corollary~\ref{cor:complexity}. 

\begin{restatable}{corollary}{corcomplexity}
\label{cor:complexity}
    Let $\Psi$ be a recursively enumerable class of counting minimal and connected conjunctive queries with at least one free variable. The problem $\#\textsc{CQ}(\Psi)$ is solvable in polynomial time if and only if the WL-dimension of $\Psi$ is bounded; the ``only if'' is conditioned under the assumption $\mathrm{FPT} \neq W[1]$.
\end{restatable}

\paragraph*{Quantum Queries and the WL dimension of counting dominating sets.}
Our main result also enables us to classify the WL-dimension of more complex queries including unions of conjunctive queries and conjunctive queries with disequalities and negations over the free variables. The statement of this classification requires the consideration of finite linear combinations of conjunctive queries (also known as \emph{quantum queries}; see Definition~\ref{def:quantumquery}). A quantum query
is of the form
    $Q = \sum_{i=1}^\ell c_i\cdot (H_i,X_i)$
where, for all $i\in[\ell]$, $c_i \in \mathbb{Q}\setminus\{0\}$. The $(H_i,X_i)$ are connected and pairwise non-isomorphic conjunctive queries where each $(H_i,X_i)$ is counting minimal and $X_i\neq \emptyset$.

It is well known~\cite{ChenM16,DellRW19} that unions of conjunctive queries, existential positive queries, and conjunctive queries with disequalities and negations over the free variables all have (unique) expressions as quantum queries, that is, the number of answers to those more complex queries can be computed by evaluation the respective quantum query according to the definition
     $|\Ans(Q,G)| := \sum_{i=1}^\ell c_i\cdot |\Ans((H_i,X_i),G)|$.
For this reason, understanding the WL-dimension of linear combination of conjunctive queries allows us to also understand the WL-dimension of more complex queries.

Defining the \emph{hereditary semantic extension width} of a quantum query $Q$, denoted by $\hsew(Q)$, as the maximum semantic extension width of its terms, we obtain the following.

\begin{restatable}{corollary}{corquantum}\label{cor:quantum_WL}
    The WL-dimension of a quantum query $Q$ is equal to $\hsew(Q)$.
\end{restatable}

As a final corollary of our main result we take a look at a concrete graph parameter, the WL-dimension of which was not known so far: 
the parameter that maps each graph $G$ to
the number of size-$k$ dominating sets in $G$. Here, a dominating set of a graph $G$ is a subset of vertices $D$ of $G$ such that each vertex of $G$ is either contained in $D$ or is adjacent to a vertex in $D$.
With an easy argument, we  show that counting dominating sets of size $k$ can be expressed as a linear combination of the $k$-star queries $(S_k,X_k)$.  Using  Theorem~\ref{thm:main_thm} and Corollary~\ref{cor:quantum_WL}, we obtain the following corollary.
\begin{corollary}\label{cor:intro_domset}
For each positive integer $k$, the WL-dimension of 
the graph parameter that maps each graph $G$ to the number of size-$k$ dominating sets in $G$   is equal to $k$.
\end{corollary}

\subsection{Discussion and Outlook}
We stated and proved our result for the case of connected conjunctive queries with at least one free variable over graphs. However, our result can easily be extended to the following.
\begin{itemize}
\item[(A)] For disconnected queries, the WL-dimension will just be the maximum of the semantic extension widths of the connected components.
\item[(B)] If no variable is free, then counting answers of a conjunctive query becomes equivalent to \emph{deciding} the existence of a homomorphism. The WL-dimension of 
the corresponding graph parameter is equal to the treewidth of the query modulo homomorphic equivalence, which 
for queries without free variables
is the same as semantic extension width. This can be proved along the lines of the analysis of Roberson~\cite{roberson2022oddomorphisms}. 
\item[(C)] Barceló et al.\ \cite{BarceloGRR21}, and Lanzinger and Barceló~\cite{LanzingerB23} have shown very recently that the WL-algorithm (and the notions of WL-equivalence and WL-dimension) readily extend from graphs to \emph{knowledge graphs}, i.e., directed graphs with vertex labels and edge labels; parallel edges with distinct labels are allowed, but self-loops are not allowed. It is not hard to see that our analysis applies to knowledge graphs as well.
\end{itemize}
Since the technical content of this paper is already quite extensive, we decided to defer the formal statement and proofs of (A)--(C) to a future journal version.

Finally, extending our results from graphs to relational structures is more tricky, since it is not known yet whether and how WL-equivalence can be characterised via homomorphism indistinguishability from structures of higher arity.\footnote{A characterisation for the special case of \emph{constant} arity $r\geq 2$ was recently established independently by Butti and Dalmau~\cite{ButtiD21}, and by Dawar, Jakl, and Reggio~\cite{DawarJR21}.} However, recent works by B{\"{o}}ker~\cite{Boker19} and by Scheidt and Schweikardt~\cite{ScheidtS23} provide first evidence that homomorphism counts from hypergraphs of bounded generalised hypertreewidth might be the right answer. We leave this for future work.

\subsection{Organisation of the Paper}
We start by introducing further necessary notation and concepts in Section~\ref{sec:prelims}. Afterwards, we prove the upper bound of the WL-dimension in Section~\ref{sec:upper_bound}, and we prove the lower bound in Section~\ref{sec:lower_bound}. Those two sections can be read independently from each other and the majority of the conceptual and technical work is done in Section~\ref{sec:lower_bound}.
Finally, we prove Theorem~\ref{thm:main_thm}, as well as its consequences, in Section~\ref{sec:main}.

\subsection*{Acknowledgements}
We are very grateful to Matthias Lanzinger for helpful advice during the early stages of this work.

\section{Preliminaries}\label{sec:prelims}
Given a set $S$, we write $\bij(S)$ for the set of all bijections from~$S$ to itself. Given a function $f:A \to B$ and a subset $X\subseteq A$, we write $f|_X: X \to B$ for the restriction of $f$ on $X$. We write $\pi_1$ for the projection of a pair to its first component, that is, $\pi_1(a,b)=a$. Given a positive integer $\ell$ we set $[\ell]=\{1,\dots,\ell\}$.

All graphs in this paper are undirected and simple (without self-loops and without parallel edges).
Given a graph $G=(V,E)$, a vertex $u\in V$ and a subset $U$ of~$V$, $N(u) = \{ v\in V \mid \{u,v\} \in E\}$ and $N(U) = \cup_{u\in U} N_u$. We say that a connected component $C$ of a graph $H$ is \emph{adjacent} to a vertex $v$ of $H$ if there is a vertex $u$ in $C$ that is adjacent to $v$. Given a subset $S$ of vertices of a graph $G$, we write $G[S]$ for the graph induced by the vertices in $S$.
 
A \emph{homomorphism} from a graph $H$ to a graph $G$ is a function $h:V(H)\to V(G)$ such that, for all edges $\{u,v\}\in E(H)$, $\{h(u),h(v)\}$ is an edge of $G$. We write $\Hom(H,G)$ for the set of all homomorphisms from $H$ to $G$. An \emph{isomorphism} from $H$ to $G$ is a bijection $b:V(H) \to V(G)$ such that, for all $u,v\in V(H)$, $\{u,v\}\in E(H)$ if and only if $\{h(u),h(v)\}\in E(G)$. We say that $H$ and $G$ are \emph{isomorphic}, denoted by $H\cong G$, if there is an isomorphism from $H$ to $G$. An \emph{automorphism} of a graph $H$ is an isomorphism from $H$ to itself, and we write $\Aut(H)$ for the set of all automorphisms of $H$.

\subsection{Conjunctive Queries}\label{sec:cqs}
As stated in the introduction, we focus on conjunctive queries on graphs. This allows us to follow the notation of~\cite{DellRW19}.

\begin{definition}
A \emph{conjunctive query} is a pair $(H,X)$ where $H$ is the underlying graph and $X$ is the set of free variables.
When~$H$ and~$X$ are clear from context we will use $Y$ to denote $V(H)\setminus X$. We say that a conjunctive query $(H,X)$ is \emph{connected} if $H$ is connected.
\end{definition}

It is well-known (see e.g. \cite{ChenM15,ChenM16,DellRW19}) that the set of answers of a conjunctive query in a graph~$G$ is the set of assignments from the free variables to the vertices of $G$ that can be extended to a homomorphism.

\begin{definition} 
Let $(H,X)$ be a conjunctive query and let $G$ be a graph. The set of answers of $(H,X)$ in $G$
is given by
$\Ans((H,X),G) = \{ a:X \to V(G)\mid \exists h\in \Hom(H,G): h|_X=a \}$. 
\end{definition}

We say that two conjunctive queries $(H_1,X_1)$ and $(H_2,X_2)$ are \emph{isomorphic}, denoted by $(H_1,X_1)\cong (H_2,X_2)$ if there is an isomorphism from $H_1$ to $H_2$ that maps $X_1$ to $X_2$. 

Throughout this work, we will focus on \emph{counting minimal} conjunctive queries. 

\begin{definition}[Counting Equivalence and Counting Minimality]\label{def:zzz}
We say that two conjunctive queries $(H_1,X_1)$ and $(H_2,X_2)$ are  \emph{counting equivalent}, denoted by $(H_1,X_1)\sim (H_2,X_2)$, if for each graph $G$,  $|\Ans((H_1,X_1),G)|=|\Ans((H_2,X_2),G)|$. A conjunctive query is said to be \emph{counting minimal} if it 
it is minimal (with respect to taking subgraphs) in its counting equivalence class.
\end{definition}

It is known that all counting minimal conjunctive queries within a counting equivalence class are  isomorphic~\cite{ChenM16,DellRW19}. 
If a query has no existential variables so that $X=V(H)$ then
 counting equivalence is the same as isomorphism.
 If all variables are quantified so that $X = \emptyset$ 
 then counting equivalence is the same as homomorphic equivalence (also called semantic equivalence).

\subsection{Treewidth and Extension Width}\label{sec:width}
We start by introducing tree decompositions and treewidth.
\begin{definition}\label{def:treewidth}
Let $H$ be a graph. A \emph{tree decomposition} of $H$ is a pair consisting of a tree $T$ and a collection of sets, called \emph{bags}, $\mathcal{B}=\{B_t\}_{t\in V(T)}$, such that the following conditions are satisfied:
    \begin{itemize}
        \item[(T1)] For all $v\in V(H)$ there is a bag $B_t$ with $v\in B_t$.
        \item[(T2)] For all $v\in V(H)$ the subgraph of $T$ induced by the vertex set $\{t\in V(T)\mid v\in B_t\}$ is connected.
        \item[(T3)] For all $e\in E(H)$, there is a bag $B_t$ with $e \subseteq B_t$.
    \end{itemize}
    The \emph{width} of $(T,\mathcal{B})$ is $\max_{t\in V(T)} |B_t|-1$ and a tree decomposition of minimum width is called \emph{optimal}.
    The \emph{treewidth} of $H$, denoted by $\tw(H)$, is the width of an optimal tree decomposition of $H$. 
    The treewidth of a conjunctive query $(H,X)$, denoted by $\tw(H,X)$, is the treewidth of~$H$.
\end{definition}

Next we introduce the extension width of a conjunctive query.

\begin{definition}[$\Gamma(H,X)$ and Extension Width]
\label{def:ew}
Let $(H,X)$ be a conjunctive query.
The \emph{extension}~$\Gamma(H,X)$ of $(H,X)$ is a graph with vertex set $V(H)$ and 
edge set $E(H)\cup E'$, where $E'$ is the set of all $\{u,v\}$ such that $u,v\in X$, $u \neq v$, and there is a connected component of~$H[Y]$ which is adjacent to both $u$ and $v$ in $H$.
The \emph{extension width} of of a conjunctive query $(H,X)$
is defined by 
$\ctw(H,X) := \tw(\Gamma(X,H))$.
\end{definition}

We will often restrict our analysis to counting minimal conjunctive queries. This requires us to lift the notion of extension width as follows.

\begin{definition}[Semantic Extension Width]
  The \emph{semantic extension width} of a conjunctive query $(H,X)$, denoted by $\sew(H,X)$, is the extension width of a counting minimal conjunctive query $(H',X')$ with $(H,X)\sim (H',X')$.
\end{definition}
Note that the semantic extension width is well-defined since all counting minimal $(H',X')$ with $(H,X)\sim (H',X')$ are isomorphic.

\subsection{The $\ell$-copy $F_\ell(H,X)$}
One of the most central operations on conjunctive queries invoked in this work is a cloning operation on existentially quantified variables, defined as follows.

\begin{definition}[$F_\ell(H,X)$]
Let $(H,X)$ be a conjunctive query and let $\ell$ be a positive integer. The $\ell$\emph{-copy}
$F_\ell(H,X)$ is defined as follows.
The vertex set of $F_\ell(H,X)$ is $X \cup (Y \times [\ell])$.
Let 
\begin{align*}
E_X &= \{\{u,v\} \in E(H) \cap X^2\},\\
E_{X,Y} &= \{ \{u,(v,i)\} \mid u\in X,  v\in Y, i\in[\ell], \{u,v\}\in E(H)\}, \textrm{ and }\\ 
E_Y&=\{\{(u,i),(v,i)\} \mid  \{u,v\} \in E(H) \cap Y^2, i\in [\ell]\}.
\end{align*}
The edge set of $F_\ell(H,X)$ is $E_X \cup E_{X,Y} \cup E_Y$.
\end{definition}

There is a natural homomorphism from $F_\ell(H,X)$ to $H$
which we denote by $\gamma[H,X,\ell]$.

\begin{definition}\label{def:gamma} 
Let $(H,X)$ be a conjunctive query and let $\ell$ be a positive integer.
Define $\gamma[H,X,\ell]\colon V(F_\ell(H,X)) \to V(H)$ as follows:
\[\gamma[H,X,\ell](u) =
\begin{cases}
    u & u \in X\\
    \pi_1(u) & u\in Y \times [\ell]
\end{cases}\]
We will just write $\gamma=\gamma[H,X,\ell]$ if $(H,X)$ and $\ell$ are clear from the context.
\end{definition}

\begin{observation}\label{obs:gamma_is_col}
The function $\gamma$ is a homomorphism from~$F_\ell(H,X)$ to~$H$.
\end{observation}

Next, we relate the treewidth of the graph $F_\ell(H,X)$ to the extension width of $(H,X)$.

\begin{lemma}\label{lem:width}
Let $(H,X)$ be a conjunctive query and let $\ell$ be a positive integer. The treewidth of $F_\ell(H,X)$ is at most $\ctw(H,X)$.
\end{lemma}
\begin{proof}
Let $\Gamma=\Gamma(H,X)$ and
let $C_1,\dots,C_m$ be the connected components of $H[Y]$.
For each $i\in [m]$, 
let $\delta_i = N(C_i) \cap X$ and
let $\hat{C}_i = C_i \cup \delta_i$. 
Since $\delta_i$ is a clique in~$\Gamma$,
there is an optimal tree decomposition $(\calT_i,\calB_i)$ of~$\Gamma[\hat{C}_i]$ with $\delta_i$ as a bag.
For $j\in[\ell]$, let $(\calT_i^j,\calB_i^j)$ be a copy of $(\calT_i,\calB_i)$ where $B_i^j$ is the bag corresponding to~$\delta_i$.

Let $(\calT_X,\calB_X)$ be an optimal tree decomposition of $\Gamma[X]$. Choose $(\calT_X,\calB_X)$ such that there is a bag $B_{X,i}$ corresponding to each~$\delta_i$.

Finally, construct a tree decomposition $(\calT,\calB)$ of $F_\ell(H,X)$ by identifying $B_{X,i}$ and $B_i^j$ for each $i\in [m]$ and $j\in [\ell]$. This tree decomposition shows that $\tw(F_\ell(H,X))\leq \tw(\Gamma)$.
\end{proof}

The following lemma follows implicitly from~\cite{Bodlaender03}. We include a proof for completeness.

\begin{lemma}\label{lem:width_hard}
    Let $(H,X)$ be a conjunctive query. There exists a positive integer $\ell$ such that $\etw(H,X)\leq \tw (F_\ell(H,X))$.
\end{lemma}
\begin{proof}
Choose any $\ell> |V(H)|+1$, and let $\gamma=\gamma[H,X,\ell]$.
Let $(T,\mathcal{B})$, with $\mathcal{B}=\{B_t\}_{t\in V(T)}$ be an optimal tree decomposition of $F_\ell(H,X)$. We 
prove the lemma by
constructing a tree decomposition $(T',\mathcal{B}')$
of $\Gamma(X,H)$ with
with width at most the width of $(T,\mathcal{B})$.
Let $T'=T$. For each $t\in V(T)$, define
$B'_t = \{ \gamma(v) \mid 
\mbox{$v\in B_t$ and $v$ is not of the form $(v,i)$ for $i>1$} \}$.

We claim that $(T',\mathcal{B}')$ is a tree-decomposition of $\Gamma(H,X)$. This claim proves the lemma, since the width of $(T',\mathcal{B}')$ is clearly at most the width of $(T,\mathcal{B})$, and since the extension width of $(H,X)$ is, by definition, the treewidth of $\Gamma(H,X)$. Hence it remains to prove our claim by establishing (T1), (T2), and (T3) from Definition~\ref{def:treewidth}.
In each case, for each $v\in V(\Gamma(H,X))$, let $v'=v$ if $v\in X$ and let $v'=(v,1)$ if $v \in Y$.

\begin{itemize}
\item[(T1)] Consider $v\in V(\Gamma(H,X))=V(H)$.  Then $v'\in V(F_\ell(H,X))$ and thus there is a bag $B_t$ with $v'\in B_t$. Since $v'\neq (v,i)$ for $i>1$ and since $\gamma(v')=v$, $v\in B'_t$.
\item[(T2)] Consider $v\in V(\Gamma(H,X))$ and let $s$ and $t$ be any pair of vertices of $T$ such that $v\in B'_s$ and $v\in B'_t$. We show that there is an $s$-$t$-path $P$ in $T$ such that $v\in B'_u$ for each $u\in P$. 
Since $v\in B'_s$ and $v\in B'_t$, we have $v' \in B_s$ and $v' \in B_t$. Using that $(T,\mathcal{B})$ is a tree-decomposition, there is a path $P$ in $T$ such that $v'\in B_u$, and thus $v\in B'_u$, for each $u\in P$.    \item[(T3)] Consider $e=\{v_1,v_2\}\in E(\Gamma(H,X))$.  Recall from Definition~\ref{def:ew} that $E(\Gamma(H,X)) = E(H)\cup E'$, where $E'$ contains all $\{u,v\}$ such that there is a connected component $C$ of $H[Y]$ that is adjacent to both $u$ and $v$ in $H$.
    
We distinguish between two cases. For the easy case, suppose that $e\in E(H)$. Then $e'=\{v'_1,v'_2\}\in E(F_\ell(H,X))$. Thus there is a bag $B_t\in \mathcal{B}$ with $e'\subseteq B_t$. Consequently, $e\in B'_t$.

For the more difficult case, suppose $e\in E'$. Then $v_1,v_2\in X$ and thus $v_1=v'_1$ and $v_2=v'_2$. Moreover, there is a connected component $C$ of $H[Y]$ that is adjacent to both $v_1$ and $v_2$. Since there are $\ell$ copies of $C$ in $F_\ell(H,X)$, there are at least $\ell$ vertex disjoint paths from $v_1$ to $v_2$ in $F_\ell(H,X)$. Now, using known separation properties of tree decompositions (see for instance \cite[Lemma 5]{DourisboureG07}), there is either a bag $B_t$ of $(T,\mathcal{B})$ that contains $v_1$ and $v_2$ --- in this case, we are done --- or there is an edge $e=\{s,t\}$ of $T$ such that $S:=B_s \cap B_t$ separates $v_1$ and $v_2$, that is, $S$ does not contain $v_1$ and $v_2$, and every $v_1$-$v_2$-path of $F_\ell(H,X)$ intersects $S$. This requires $|S|\geq \ell$ and thus $|B_s|\geq \ell$ (and $|B_t| \geq \ell$). Consequently, using the fact that $(T,\mathcal{B})$ is optimal, the treewidth of $F_\ell(H,X)$ is at least $\ell> |V(H)|+1$, which yields a contradiction, since by Lemma~\ref{lem:width} we have \[\tw(F_\ell(H,X))\leq \ctw(H,X) = \tw(\Gamma(H,X))\leq |V(\Gamma(H,X))| = |V(H)| \,.\]
Thus $v_1$ and $v_2$ must both be contained in some bag $B_t$, and thus also in the bag $B'_t$.
\end{itemize}
With (T1)-(T3) established, the proof is concluded.
\end{proof}

In combination,  Lemmas~\ref{lem:width} and~\ref{lem:width_hard} provide an alternative characterisation of the extension width, which we will be using for the remainder of the paper.

\begin{corollary}\label{cor:ewidth_alternative}
Let $(H,X)$ be a conjunctive query. Then
$ \ctw(H,X)=\max\{\tw(F_\ell(H,X))\mid \ell \in \mathbb{Z}_{>0}\}$.
\end{corollary}
\begin{proof}
    The corollary follows immediately from Lemmas~\ref{lem:width} and~\ref{lem:width_hard}.
\end{proof}

\subsection{Weisfeiler-Leman Equivalence, Invariance and Dimension}

In order to make this work self-contained, we will use the characterisation of Weifeiler-Leman (from now on just ``WL'') equivalence via homomorphism indistinguishability due to Dvor{\'{a}}k~\cite{Dvorak10} and Dell, Grohe and Rattan~\cite{DellGR18}. We recommend the survey of Arvind for a short but comprehensive introduction to the classical definition using the WL-algorithm~\cite{Arvind16}.

\begin{definition}[WL-Equivalence]\label{def:wl_equivalence}
    Let $k$ be a positive integer. Two graphs $G$ and $G'$ are $k$\emph{-WL-equivalent}, denoted by $G\cong_k G'$, if for every graph $H$ of treewidth at most $k$ we have $|\Hom(H,G)|=|\Hom(H,G')|$.
\end{definition}

Note that WL-equivalence is monotone in the sense that $G\cong_k G'$ implies 
that for every $k'\leq k$,
$G\cong_{k'} G'$. 
A graph parameter $f$ is called $k$\emph{-WL-invariant} if, 
for every pair of graphs $G,G'$ with $G \cong_k G'$,
$f(G)=f(G')$. Observe that, for $k \geq k'$, every $k'$ -WL-invariant graph parameter is also $k$-WL-invariant.
Thus, following the definition of Arvind et al.\ \cite{ArvindFKV22}, we define the \emph{WL-dimension} of a graph parameter $f$ as the minimum $k$ for which $f$ is $k$-WL-invariant, if such a $k$ exists, and $\infty$ otherwise.

\begin{definition}[WL-dimension of conjunctive queries]\label{def:WLdim}
Let $(H,X)$ be a conjunctive query. The \emph{WL-dimension} of $(H,X)$ is the WL-dimension of the function $G \mapsto |\Ans((H,X),G)|$. 
\end{definition}

\section{Upper Bound on the WL-Dimension}\label{sec:upper_bound}

The goal of this section is to prove the following upper bound.

\begin{theorem}\label{thm:upper_bound}
    Let $(H,X)$ be a conjunctive query. Then the WL-dimension of $(H,X)$ is at most $\ctw(H,X)$.
\end{theorem}
For the proof of Theorem~\ref{thm:upper_bound} we will use the following interpolation argument. 
\begin{lemma}\label{lem:gen}
Let $(H,X)$ be a conjunctive query.
Let $G_1$ and $G_2$ be graphs.
Suppose that, for all positive integers $\ell$,  $|\Hom(F_\ell(H,X),G_1)| = |\Hom(F_\ell(H,X),G_2)|$. Then  $|\Ans((H,X) ,G_1)| = |\Ans((H,X),G_2)|$.
\end{lemma}
\begin{proof}
Let $G$ be a graph and let $\sigma\colon X \to V(G)$. Define
\[\Ext(\sigma)=\{\rho\colon Y \to V(G) \mid \sigma \cup \rho \in \Hom(H,G) \}.\]
Let $\Omega$ be the set of functions from~$Y$ to~$V(G)$ and consider any $\Upsilon \subseteq \Omega$.  Define
\begin{align*}
    H^G(\Upsilon) &= \{ h\in \Ans((H,X),G)\mid \Ext(h)= \Upsilon \}\\ 
    \hat{H}_\ell^G(\Upsilon) &= \{ h\in \Hom(F_\ell(H,X),G)\mid \Ext(h|_{X}) = \Upsilon \} 
\end{align*} 
First observe that for any $\Upsilon\subseteq \Omega$, $|\hat{H}_\ell^G(\Upsilon)| = |H^G(\Upsilon)|\cdot |\Upsilon|^\ell$.
Moreover, 
\begin{align*}
    |\Ans((H,X),G)| & = \sum_{\emptyset \neq \Upsilon \subseteq \Omega} |H^G(\Upsilon)|, \mbox{and}\\
    |\Hom(F_\ell(H,X),G)| & = \sum_{\emptyset \neq \Upsilon \subseteq \Omega} |\hat{H}^G_\ell(\Upsilon)|.
\end{align*}
Now let $G_1$ and $G_2$ be graphs with $|\Hom(F_\ell(H,X),G_1)| = |\Hom(F_\ell(H,X),G_2)|$ for all positive integers $\ell$. 
Let $\Omega_1$ be the set of functions from~$Y$ to~$V(G_1)$ and let $\Omega_2$ be the set of functions from~$Y$ to~$V(G_2)$.
Let $\hat{n}=\max\{|\Omega_1|,|\Omega_2|\}$. For every positive integer~$\ell$,
\begin{align*}
    ~~&|\Hom(F_\ell(H,X),G_1)| = |\Hom(F_\ell(H,X),G_2)| \\
    \Leftrightarrow & \sum_{\emptyset \neq \Upsilon \subseteq \Omega_1} |\hat{H}^{G_1}_\ell(\Upsilon)| - \sum_{\emptyset \neq \Upsilon \subseteq \Omega_2} |\hat{H}^{G_2}_\ell(\Upsilon)| = 0\\
    \Leftrightarrow & \sum_{\emptyset \neq \Upsilon \subseteq \Omega_1} |{H}^{G_1}(\Upsilon)| \cdot |\Upsilon|^\ell - \sum_{\emptyset \neq \Upsilon \subseteq \Omega_2)} |{H}^{G_2}(\Upsilon)|\cdot |\Upsilon|^\ell = 0 \\
    \Leftrightarrow & \sum_{i=1}^{\hat{n}} i^\ell \cdot \Bigg(\sum_{\substack{ \Upsilon \subseteq \Omega_1\\ |\Upsilon|=i}} |{H}^{G_1}(\Upsilon)| - \sum_{\substack{ \Upsilon \subseteq \Omega_2\\ |\Upsilon|=i}} |{H}^{G_2}(\Upsilon)| \Bigg) = 0 
\end{align*}
Note that this yields a system of linear equations. For each positive integer~$\ell$, we have the equation
$ \sum_{i=1}^{\hat{n}} c_i \cdot i^\ell =0$
where
\[c_i = \Bigg(\sum_{\substack{ \Upsilon \subseteq \Omega_1\\ |\Upsilon|=i}} |{H}^{G_1}(\Upsilon)| - \sum_{\substack{ \Upsilon \subseteq \Omega_2\\ |\Upsilon|=i}} |{H}^{G_2}(\Upsilon)| \Bigg). \]
The  matrix corresponding to this system of
equations is a Vandermonde matrix, so it is invertible. Thus $c_i = 0$ for all $i\in \{1,\dots,n\}$. Therefore
\begin{align*}
    |\Ans((H,X),G_1)| & = \sum_{\emptyset \neq \Upsilon \subseteq \Omega_1} |H^{G_1}(\Upsilon)| = \sum_{i=1}^{\hat{n}} \sum_{\substack{\Upsilon\subseteq \Omega_1\\ |\Upsilon|=i}} |{H}^{G_1}(\Upsilon)| \\
    ~&= \sum_{i=1}^{\hat{n}} \sum_{\substack{\Upsilon\subseteq \Omega_2\\ |\Upsilon|=i}} |{H}^{G_2}(\Upsilon)| = \sum_{\emptyset \neq \Upsilon \subseteq \Omega_2} |H^{G_2}(\Upsilon)|=|\Ans((H,X),G_2)|.
\end{align*}
\end{proof}

The proof of Lemma~\ref{lem:gen} immediately implies the following observation, Observation~\ref{obs_for_intro}. Note that the graphs $F_\ell(H,X)$ that are referred to in Lemma~\ref{lem:gen}
have treewidth at most $\etw(H,X)$ by Lemma~\ref{lem:width}. In Observation~\ref{obs_for_intro} there are two possibilities. If we start with a query $(H,X)$ that is counting minimal, we can apply directly the proof of Lemma~\ref{lem:gen}. Otherwise, we apply the proof of Lemma~\ref{lem:gen} to a counting-equivalent counting-minimal query. 
\begin{observation}\label{obs_for_intro}
    Let $(H,X)$ be a conjunctive query of semantic extension width $k$ and let $G$ be a graph. There is a finite sequence of graphs $F_1,\dots,F_n$ of treewidth at most $k$, such that $|\Ans((H,X),G)|$ can be computed via Gaussian elimination from the homomorphism counts $|\Hom(F_\ell,G)|$ for $\ell\in \{1,\dots,n\}$.
\end{observation}

\begin{proof}[Proof of Theorem~\ref{thm:upper_bound}]
Let $(H,X)$ be a conjunctive query.
Let $k=\ctw(H,X)$. 
We wish to show that the WL-dimension of $(H,X)$ is at most~$k$ which is equivalent to showing that
the function $G \mapsto |\Ans((H,X),G)|$ is $k$-WL invariant.
To do this, we
show that, for any pair of graphs $G$ and $G'$ with $G \cong_k G'$,  $|\Ans((H,X),G)|=|\Ans((H,X),G')|$.

Consider $G$ and $G'$ with $G \cong_k G'$.
This implies that for every graph $H$ with treewidth at most $k$,  
$|\Hom(H,G)|=|\Hom(H,G')|$. 
From the definition of $\ctw(H,X)$ and 
Corollary~\ref{cor:ewidth_alternative},  
for every positive integer~$\ell$, 
the treewidth of $F_\ell(H,X)$ is at most~$k$.
Thus,
$|\Hom(F_\ell(H,X),G)|=|\Hom(F_\ell(H,X),G')|$. 
The claim then follows directly by Lemma~\ref{lem:gen}.
\end{proof}

\section{Lower Bound on the WL-Dimension}\label{sec:lower_bound}
The goal of this section, which is the technical heart of the paper, is the proof of the following lower bound.

\begin{restatable}{theorem}{thmlowerbound}\label{thm:lower_bound} 
Let $(H,X)$ be a counting minimal conjunctive query such that $H$ is connected, and $ \emptyset \subsetneq X \subsetneq V(H) $. Then the WL-dimension of $(H,X)$ is at least $\ctw(H,X)$.
\end{restatable}

In order to prove Theorem~\ref{thm:lower_bound}, we will find graphs $G$ and $G'$ such that $G\cong_{k-1} G'$, where $k=\ctw(H,X)$, and $|\Ans((H,X),G)|\neq |\Ans((H,X),G')|$. As explained in the introduction, we will rely on a recently developed version of the CFI graphs of Cai, F\"urer and Immerman~\cite{CaiFI92}. The following subsection will provide a concise and self-contained explanation of the construction and properties of CFI graphs.

\subsection{CFI Graphs}
We start with a formal definition of a well-known version of CFI graphs~\cite{Furer01} (see also~\cite{roberson2022oddomorphisms}).

\begin{definition}[CFI graphs, $\chi(G,W)$] \label{def:B}
Let $G$ be a graph and let $W$ be a subset of~$V(G)$. For every vertex $w$ of~$G$,
let $\delta_{w,W} = |\{w\} \cap W|$. The graph $\chi(G,W)$ is defined as follows.
The vertex set is $V(\chi(G,W)) := \{ (w,S) \mid w\in V(G), S \subseteq N_G(w), \delta_{w,W} \equiv |S| \pmod 2\}$.
The edge set is $$ E(\chi(G,W)) := \{ \{ (w,S), (w',S') \} \mid 
\mbox{$\{w,w'\} \in E(G)$ and
$w'\in S \Longleftrightarrow w\in S'$} \}.$$
\end{definition}

For any fixed $G$, the isomorphism class of $\chi(G,W)$ depends only on the parity of $|W|$:
\begin{lemma} [Lemma 3.2 in \cite{roberson2022oddomorphisms}]\label{lem:iso} \label{lem:cfi_propsB}
Let $G$ be a connected graph 
and let $W,W' \subseteq V(G)$. Then $\chi(G,W) \cong \chi(G,W')$ if
and only if   $|W|\equiv |W'|\pmod 2$.
\end{lemma}

\noindent Neuen~\cite{neuen2023homomorphism} established the following WL-equivalence result for $\chi(G,\emptyset)$ and $\chi(G,\{w\})$.
\begin{lemma}[Theorem 4.2, Lemma 4.4 and Theorem 5.1 in \cite{neuen2023homomorphism}]\label{lem:cfi_propsA}
Let $G$ be a graph of treewidth~$t$ and let $w$ be a vertex of~$G$.
Then for all $k<t$,  $\chi(G,\emptyset)\cong_k\chi(G,\{w\})$.
\end{lemma}

\subsection{Cloning Vertices in CFI Graphs}\label{sec:aaa}
We will introduce some notions and properties of coloured graphs and of CFI graphs due to Roberson~\cite{roberson2022oddomorphisms}. First of all, since we will work with vertex colourings induced by homomorphisms throughout this section, we adopt the well-established notion of $H$-\emph{colourings} of graphs.
\begin{definition}
We refer to a homomorphism from a graph~$G$ to a graph~$H$ as an $H$-\emph{colouring} of~$G$.
\end{definition}
 
Recall that $\pi_1$ is the projection that maps a pair $(a,b)$ to the first component $a$.

\begin{observation} [\cite{roberson2022oddomorphisms}]
\label{obs:rob}
Let $F$ be a graph
and let $W$ be a subset of $V(F)$.
The function $\pi_1$  
is a homomorphism from $\chi(F,W)$ to $F$.
\end{observation}

Technically, to get the $F$-colouring of $\chi(F,W)$ in Observation~\ref{obs:rob}, one should restrict $\pi_1$ to the domain $V(\chi(F,W))$, but it will not be important to capture this in our notation.
The following is an extension of a notion introduced in~\cite[Section 3.1]{roberson2022oddomorphisms} from CFI graphs to coloured graphs.
\begin{definition}\label{def:part_homs_tau}
Let $H$, $G$, and $F$ be graphs, let $c$ be a homomorphism from $G$ to $F$, and
let $\tau$ be a homomorphism from $H$ to~$F$.  We define
$ \Hom_\tau(H,G,F,c) = \{h \in \Hom(H,G)\mid c(h(\cdot)) = \tau \}$.
\end{definition}

\begin{figure}
    \centering
\begin{tikzcd}[column sep=4em, row sep=4em]  
    H \arrow[r, "{\scalebox{1.25}{\(h\)}}"] \arrow[rd, swap, "{\scalebox{1.25}{\(\tau\)}}"] & G \arrow[d, "{\scalebox{1.25}{\(c\)}}"] \\
    & F
\end{tikzcd}
    \caption{Each homomorphism $h$ from $H$ to $G$ induces a homomorphism $\tau$ from $H$ to $F$ by composing $h$ with the $F$-colouring $c$ of $G$. By partitioning $\Hom(H,G)$ along the induced homomorphisms to $F$, we obtain  Observation~\ref{obs:partition}.} 
    \label{fig:cd_1}
\end{figure}
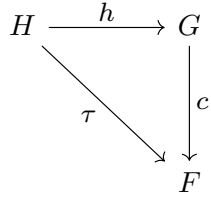

\begin{observation}\label{obs:partition}
Let $H$, $G$, and $F$ be graphs and let 
$c$ be a homomorphism from $G$ to $F$. Then
\[ |\Hom(H,G)| = \sum_{\tau\in \Hom(H,F)} |\Hom_\tau(H,G,F,c)|.\]
\end{observation}

\begin{theorem}[Theorem 3.6 in~\cite{roberson2022oddomorphisms}]\label{thm:roberson_main_bound}
Let $H$ be a graph, let $F$ be a connected graph, let $W\subseteq V(F)$, and let $\tau\in \Hom(H,F)$. Then
$|\Hom_\tau(H,\chi(F,W),F,\pi_1)| \leq |\Hom_\tau(H,\chi(F,\emptyset),F,\pi_1)|$.
\end{theorem}

\begin{definition}[Cloning Colour-Blocks]\label{def:cloning_operation}
Let $G$ be a graph,
let $F$ be a connected graph, and let $c$ be a homomorphism from $G$ to~$F$. Let
$k$ be a positive integer, let  $\vv=(v_1,\dots,v_k)$ be a $k$-tuple of pairwise distinct vertices of $F$, and let $\vz=(z_1,\dots,z_k)$ be $k$-tuple of positive integers. The graph $\newgraph(G,F,c,\vv,\vz)$ 
is obtained from $G$ by cloning, for each $i\in[k]$, the colour class of $v_i$ under~$c$ precisely $z_i-1$ times. Formally, 
for each $v\in V(F)$, 
let $B_v=c^{-1}(v)$. The vertices of 
$\newgraph(G,F,c,\vv,\vz)$ are  
\[ \bigcup_{u\in V(F)\setminus \vv} B_u \cup \bigcup_{i\in[k]} (B_{v_i} \times \{1,\dots,z_i\}).\]
The vertices contained in $\bigcup_{u\in V(F)\setminus \Vec{v}} B_u$ are called \emph{primal vertices} and the vertices contained in $\bigcup_{i\in[k]} (B_{v_i} \times \{1,\dots,z_i\})$ are called \emph{cloned vertices}.
Two vertices $x$ and $y$ of 
$\newgraph(G,F,c,\vv,\vz)$ are adjacent if and only if
\begin{itemize}
\item $x$ and $y$ are primal vertices, and $\{x,y\} \in E(G)$, or
\item $x$ is primal, $y$ is a clone, and $\{x,\pi_1(y)\} \in E(G)$, or
\item $x$ is a clone, $y$ is primal, and $\{\pi_1(x),y\} \in E(G)$, or
\item $x$ and $y$ are clones, and $\{\pi_1(x),\pi_1(y)\} \in E(G)$.
\end{itemize}
We define a function 
$\newcol(G,F,c,\vv,\vz) \colon V(\newgraph(G,F,c,\vv,\vz) \to V(F)$ 
by mapping primal vertices $u$ to $c(u)$, and cloned vertices $(u,i)$ to $c(u)$. It is easy to see that 
$\newcol(G,F,c,\vv,\vz)$   is 
a homomorphism from  
$\newgraph(G,F,c,\vv,\vz)$ to~$F$. 
\end{definition}

\begin{lemma}\label{lem:cloning_property}
Let $H$ and $G$ be  graphs and
let $F$ be a connected graph. Let $c$ be a homomorphism from $G$ to~$F$
and let $\tau$ be a homomorphism from~$H$ to~$F$.  Let $\vv=(v_1,\dots,v_k)$ be a $k$-tuple of distinct vertices of $F$ and let $\vz=(z_1,\dots,z_k)$ be a $k$-tuple of positive integers. For all $i\in[k]$, let $d_i$ be the number of vertices of $H$ that are mapped by $\tau$ to $v_i$, i.e., $d_i = |\{ u\in V(H) \mid \tau(u)=v_i \}|$. 
Let $G' = \newgraph(G,F,c,\vv,\vz)$ and let $c' = \newcol(G,F,c,\vv,\vz)$.
Then  
\[ |\Hom_\tau(H, G',F,c')|= 
|\Hom_\tau(H,G,F,c)| \cdot \prod_{i=1}^k z_i^{d_i}. \]\end{lemma}
\begin{proof}
Let $\rho\colon V(G') \to V(G)$ be the function that maps cloned vertices to their original counterparts, that is
\[ \rho(x) = \begin{cases}
x & \mbox{if $x$  is a primal vertex} \\
\pi_1(x) & \mbox{otherwise}  
\end{cases} \]
Observe that $\rho$ is a homomorphism from~$G'$ to~$G$.
We define an equivalence relation on the set 
$\Hom_\tau(H,G',F,c')$   by setting $h \sim h'$ if and only if 
$\rho(h(\cdot)) = \rho(h'(\cdot))$. 
For every $h\in \Hom_\tau(H,G',F,c')$,
$\rho(h(\cdot))$ is a homomorphism from~$H$ to~$G$   since it is the composition of homomorphisms from~$H$ to~$G'$ and from~$G'$ to~$G$. 
Moreover, since $h \in \Hom_\tau(H,G',F,c')$,
$c'(h(\cdot)) = \tau$. 
From the definitions of $\rho$ and $c'$ it is immediate that 
$c(\rho(\cdot))   =  c'$.
Thus $c(\rho(h(\cdot)) =   c'(h(\cdot)) = \tau$,  proving that 
$\rho(h(\cdot)) \in \Hom_\tau(H,G,F,c)$; consider Figure~\ref{fig:cd_2} for an illustration.
Consequently, we can represent each equivalence class of $\sim$ by a homomorphism $\hat{h} \in \Hom_\tau(H,G,F,C)$. 

Finally, each equivalence class  has size $\prod_{i=1}^k z_i^{d_i}$, since for each $i\in[k]$, there are $z_i$ possibilities for each of the $d_i$ vertices $u\in V(H)$ with $\tau(u) = v_i$. \end{proof}

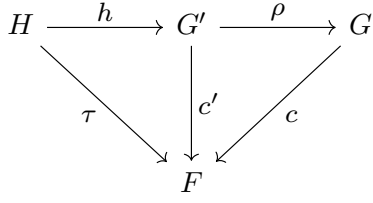
\begin{figure}
    \centering
\begin{tikzcd}[column sep=4em, row sep=4em]
    H \arrow[r, "{\scalebox{1.25}{\(h\)}}"] \arrow[rd, swap, "{\scalebox{1.25}{\(\tau\)}}"] & G' \arrow[r, "{\scalebox{1.25}{\(\rho\)}}"] \arrow[d, "{\scalebox{1.25}{\(c'\)}}"] & G \arrow[ld, "{\scalebox{1.25}{\(c\)}}"] \\
    & F 
\end{tikzcd}
    \caption{Illustration for the proof of Lemma~\ref{lem:cloning_property}: $G'=\newgraph(G,F,c,\vv\,\vz)$ is the graph obtained from $G$ by cloning vertices (Definition~\ref{def:cloning_operation}), and $\rho$ is the homomorphism from $G'$ to $G$ that maps each cloned vertex in $G'$ to its primal vertex in $G$. Moreover, $c$ is the $F$-colouring of $G$ and $c'=\newcol(G,F,c,\vv,vz)$ is, by Definition~\ref{def:cloning_operation}, the composition of $c$ and $\rho$, i.e., each cloned vertex is mapped by $c'$ to the colour of its primal vertex.} 
    \label{fig:cd_2}
\end{figure}

Next we show that cloning colour-blocks in CFI graphs preserves WL-equivalence:
\begin{lemma}\label{lem:main_cloning_lemma}
Let $F$ be a connected graph of treewidth $t+1$, let $W\subseteq V(F)$, let $\Vec{v}=(v_1,\dots,v_k)$ be a $k$-tuple of distinct vertices of $F$, and let $\Vec{z}=(z_1,\dots,z_k)$ be a $k$-tuple of positive integers. Then
$
\newgraph(\chi(F,\emptyset),F,\pi_1,\vv,\vz) 
\cong_{t}  \newgraph(\chi(F,W),F,\pi_1,\vv,\vz)$.
\end{lemma}
\begin{proof}

Let $G_\emptyset = \newgraph(\chi(F,\emptyset),F,\pi_1,\vv,\vz)$
and $G_W = \newgraph(\chi(F,W),F,\pi_1,\vv,\vz)$.
The lemma will follow immediately by Definition of WL-equivalence (Definition~\ref{def:wl_equivalence})
once
we show that, for each graph $H$ of treewidth at most $t$,
$|\Hom(H,G_\emptyset)| = |\Hom(H,G_W)|$.
To this end, fix any graph $H$ with treewidth at most $t$. 
By Lemma~\ref{lem:cfi_propsA},
$\chi(F,\emptyset) \cong_t \chi(F,W)$ since $F$ has treewidth $t+1$.
Again, by Definition~\ref{def:wl_equivalence}, we have
$|\Hom(H,\chi(F,\emptyset))| = |\Hom(H,\chi(F,W))|$. 
By Observation~\ref{obs:rob}, $\pi_1$ is a homomorphism from $\chi(F,\emptyset)$ to~$F$ and from $\chi(F,W)$ to~$F$.
By Observation~\ref{obs:partition}, 
\[ \sum_{\tau \in \Hom(H,F)} |\Hom_\tau(H,\chi(F,\emptyset),F,\pi_1)| = \sum_{\tau \in \Hom(H,F)} |\Hom_\tau(H,\chi(F,W),F,\pi_1). \]
Combining this with Theorem~\ref{thm:roberson_main_bound}, 
we find that for all $\tau \in \Hom(H,F)$, 
\begin{equation}\label{eq:cfi_extension}
|\Hom_\tau(H,\chi(F,\emptyset),F,\pi_1))| =  |\Hom_\tau(H,\chi(F,W),F,\pi_1)|.
    \end{equation}

Therefore
\begin{align*}
|\Hom(H,G_\emptyset)|  
&= \sum_{\tau \in \Hom(H,F)} |\Hom_\tau(H,G_\emptyset,F,\pi_1)| \tag{Observation~\ref{obs:partition}} \\
&= \sum_{\tau \in \Hom(H,F)}  |\Hom_\tau(H,\chi(F,\emptyset),F,\pi_1)| \cdot \prod_{i=1}^k z_i^{d_i}  \tag{Lemma~\ref{lem:cloning_property}}\\
&= \sum_{\tau \in \Hom(H,F)}  |\Hom_\tau(H,\chi(F,W),F,\pi_1)| \cdot \prod_{i=1}^k z_i^{d_i}   \tag{Equation~(\ref{eq:cfi_extension})}\\
&= \sum_{\tau \in \Hom(H,F)} |\Hom_\tau(H,G_W,F,\pi_1)|  \tag{Lemma~\ref{lem:cloning_property}} \\
&=|\Hom(H,G_W)|, \tag{Observation~\ref{obs:partition}}
\end{align*}
where the quantities $d_1,\ldots,d_k$ depend on~$\tau$ as in Lemma~\ref{lem:cloning_property}.
This concludes the proof.
\end{proof}

\subsection{Reduction to the Colourful Case}
Throughout this section, we fix a conjunctive query $(H,X)$ such that $H$ is connected, and with $X= \{x_1,\dots,x_k \}$. Fix also a positive integer $\ell$. Let $Y=V(H)\setminus X$ and $F=F_\ell(H,X)$.

Recall that the vertices of $F$ are $X \cup (Y \times [\ell])$. We start by extending the partitioning result from  Section~\ref{sec:aaa} from graphs to conjunctive queries. Moreover, recall the definition of the mapping $\gamma\colon V(F) \to V(H)$ (Definition~\ref{def:gamma}):

\[\gamma(u) =
\begin{cases}
    u & u \in X\\
    \pi_1(u) & u\in Y \times [\ell]
\end{cases}\]
Recall also that by Observation~\ref{obs:gamma_is_col}, the function $\gamma$ is a homomorphism from~$F$ to~$H$.

\begin{definition}\label{def:partition_cq}
Let $G$ be a graph, let $c$ be an $H$-colouring of $G$, and
let $\tau$ be a function from~$X$ to~$V(H)$. Define
\[\Ans^\tau((H,X),(G,c)) = \{h \in \Ans((H,X),G)\mid c(h(\cdot)) = \tau(\cdot)  \} \}\]
Let $\hat{c}$ be an $F$-colouring of~$G$. 
Define 
\[\Ans^{\tau}((H,X),(G,\hat{c})) = \{h \in \Ans((H,X),G)\mid \gamma(\hat{c}(h(\cdot)) = \tau(\cdot)  \} \}\]
\end{definition}

\begin{figure}
    \centering
\begin{tikzcd}[column sep=4em, row sep=4em]
    V(F) \arrow[r, "{\scalebox{1.25}{\(\gamma\)}}"]  & V(H) & X \arrow[ld, "{\scalebox{1.25}{\(h\)}}"] \arrow[l, swap, "{\scalebox{1.25}{\(\tau\)}}"]\\
    & V(G) \arrow[lu, "{\scalebox{1.25}{\(\hat{c}\)}}"] \arrow[u, "{\scalebox{1.25}{\(c\)}}"]
\end{tikzcd}
    \caption{Illustrations of the mappings used in Definition~\ref{def:partition_cq}.} 
    \label{fig:cd_3}
\end{figure}
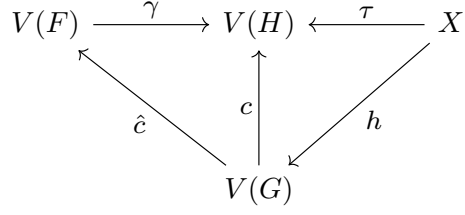

\begin{observation}\label{obs:partition_and_projection}
Let $(G,c)$ be an $H$-coloured graph. Then
\[ |\Ans((H,X),G)| = \sum_{\tau: X \to V(H)} |\Ans^\tau((H,X),(G,c))|. \]
\end{observation}

We will adapt Lemma~\ref{lem:cloning_property} to conjunctive queries in the following way.

\begin{lemma}\label{lem:cloning_property_projections}
Let $G$ be a graph and let $c$ be an $F$-colouring of $G$, and assume that $F$ is connected.
Let $\tau \colon X \to V(H)$. Let $\vv = (x_1,\ldots,x_k)$. Let $\vz=(z_1,\dots,z_k)$ be a $k$-tuple of positive integers. For all $i\in[k]$ let $d_i$ be the number of vertices of $H$ that are mapped by $\tau$ to $x_i$, i.e., $d_i := |\{ u\in X \mid \tau(u)=x_i \}|$. 
Let $G' = \newgraph(G,F,c,\vv,\vz)$ and 
$c' = \newcol(G,F,c,\vv,\vz)$.
Then  
\[ |\Ans^\tau((H,X), (G',c')| =  
|\Ans^\tau((H,X),(G,c))| \cdot \prod_{i=1}^k z_i^{d_i}. \]
\end{lemma}
\begin{proof}
We encourage the reader to consider Figure~\ref{fig:cd_4} for keeping track of the mappings and homomorphisms used in the proof.
 Let $\rho\colon V(G') \to V(G)$ be the function that maps cloned vertices to their original counterparts, that is
 \[ \rho(w) = \begin{cases}
w & \mbox{if $w\in Y$   } \\
\pi_1(w) & \mbox{otherwise}  
\end{cases} \]
Observe that $\rho$ is a homomorphism from~$G'$ to~$G$.
We define an equivalence relation on the set 
$\Ans^\tau((H,X),(G',c'))$ by setting $h \sim h'$ if and only if 
$\rho(h(\cdot)) = \rho(h'(\cdot))$. 
For every $h\in \Ans^\tau((H,X),(G',c'))$, we have $\rho(h(\cdot)) \in \Ans((H,X),G)$.

Moreover, since $h \in \Ans^\tau((H,X),(G',c'))$,
$\gamma(c'(h(\cdot))) = \tau$. 
From the definitions of $\rho$ and $c'$ it is immediate that 
$\gamma(c(\rho(\cdot)))   =  \gamma(c'(\cdot))$.
Thus $\gamma(c(\rho(h(\cdot))) =   \gamma(c'(h(\cdot))) = \tau$,  proving that 
$\rho(h(\cdot)) \in \Ans^\tau((H,X),(G,c))$.
Consequently, we can represent each equivalence class of $\sim$ by a homomorphism $\hat{h} \in \Ans^\tau((H,X),(G,c))$. 

Finally, each equivalence class  has size $\prod_{i=1}^k z_i^{d_i}$, since for each $i\in[k]$, there are $z_i$ possibilities for each of the $d_i$ vertices $u\in V(H)$ with $\tau(u) = v_i$. 
\end{proof}

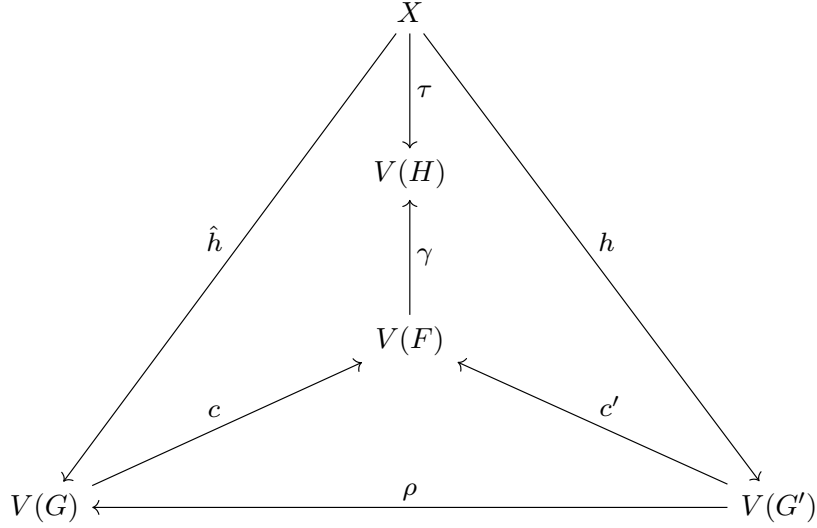
\begin{figure}
    \centering
\begin{tikzcd}[column sep=4em, row sep=4em]
~&~&X \arrow[llddd, swap, "{\scalebox{1.25}{\(\hat{h}\)}}"] \arrow[rrddd, "{\scalebox{1.25}{\(h\)}}"] \arrow[d, "{\scalebox{1.25}{\(\tau\)}}"] \\
~&~&V(H)\\
~&~&V(F) \arrow[u, swap, "{\scalebox{1.25}{\(\gamma\)}}"]\\
V(G) \arrow[urr, "{\scalebox{1.25}{\(c\)}}"] & ~& ~&~& V(G') \arrow[ull, swap, "{\scalebox{1.25}{\(c'\)}}"] \arrow[llll, swap, "{\scalebox{1.25}{\(\rho\)}}"]
\end{tikzcd}
    \caption{Overview of the mappings and homomorphisms used in the proof of Lemma~\ref{lem:cloning_property_projections}.} 
    \label{fig:cd_4}
\end{figure}

Lemma~\ref{lem:isolate_colourful} is the main technical lemma in this section, which is concerned with $\chi(F,\emptyset)$ and $\chi(F,\{x_1\})$.
Recall that the projection $\pi_1$ is an $F$-colouring of $\chi(F,\emptyset)$ and $\chi(F,\{x_1\})$. Moreover $\gamma$ is a homomorphism from $F$ to $H$. Thus:
\begin{observation}\label{obs:gamma_pi1_Hcol}
    $\gamma(\pi_1(\cdot))$ is an $H$-colouring of both $\chi(F,\emptyset)$ and $\chi(F,\{x_1\})$. 
\end{observation}

Recall that $\bij(X)$ is the set of all bijections from $X$ to $X$.

\begin{lemma}\label{lem:isolate_colourful} 
Let $c = \gamma(\pi_1(\cdot))$.
Let $\vv = (x_1,\ldots,x_k)$.
Suppose that
\[ \sum_{\tau \in \bij(X)} 
|\Ans^\tau((H,X),(\chi(F,\emptyset),c  ))| - 
|\Ans^\tau((H,X),(\chi(F,\{x_1\}),c))|  \neq 0.\]
Then there is a $k$-tuple $\vz=(z_1,\dots,z_k)$ of positive integers such that
\[ |\Ans((H,X),
\newgraph( \chi(F,\emptyset),F, c,\vv,\vz) )|  \neq |\Ans((H,X),
\newgraph( \chi(F,\{x_1\}), F,c,\vv,\vz) )|  .\]
\end{lemma}
\begin{proof}
Let $G_0 = \newgraph( \chi(F,\emptyset), c,\vv,\vz) ) $,
$c_0 = \newcol( \chi(F,\emptyset),F, c,\vv,\vz) )$,
$G_1 = \newgraph( \chi(F,\{x_1\}),F, c,\vv,\vz) ) $, and
$c_1 = \newcol( \chi(F,\{x_1\}),F, c,\vv,\vz) )$,

Suppose for contradiction that, for every $k$-tuple $\vz$ of positive integers, 
$ |\Ans((H,X), G_0)|    = |\Ans((H,X), G_1)|$. By 
Observation~\ref{obs:partition_and_projection}
$$\sum_{\tau\colon X \to V(H)} \Big(|\Ans^\tau ((H,X),(G_0,c_0))| - 
  |\Ans^\tau ((H,X),(G_1,c_1))\Big)| =0.$$
Let $d^\tau_i := |\{ u\in X \mid \tau(u)=x_i \}|$. Applying Lemma~\ref{lem:cloning_property_projections},
\begin{align}\label{eq:poly_interpol_setup}
\sum_{\tau\colon X \to V(H)}  
\Big( |\Ans^\tau ( (H,X) , (\chi(F,\emptyset),c) )| - 
    |\Ans^\tau  ( (H,X),  (\chi(F,\{x_1\}),c))| \Big) 
 \prod_{i=1}^k z_i^{d^\tau_i}     
     = 0.
\end{align} 

Define $b_\tau := |\Ans^\tau((H,X),(\chi(F,\emptyset),c))| - |\Ans^\tau((H,X),(\chi(F,\{x_1\}),c))|$. 
Then, treating the $z_i$ as variables, we can define a $k$-variate polynomial
\begin{align*}
    P(z_1,\dots,z_k) := \sum_{\tau: X \to V(H)} b_\tau \cdot \prod_{i=1}^k z_i^{d^\tau_i}\,.
\end{align*}
By~\eqref{eq:poly_interpol_setup}, $P(z_1,\dots,z_k)=0$ for all $k$-tuples of positive integers $(z_1,\dots,z_k)$. We wish to apply polynomial interpolation, which requires us to collect coefficients for all monomials: To this end, observe first that all $d^\tau_i$ are bounded from above by $k$ since $|X|=k$. Let $\mathcal{A}_k$ be the set of all $k$-tuples of integers $(a_1,\dots,a_k)$ with $0 \leq a_i \leq k$ for each $i\in[k]$.
Given a $k$-tuple $\Vec{a}\in \mathcal{A}_k$, we write $\mathsf{coeff}(\Vec{a})$ for the coefficient of the monomial $\prod_{i=1}^k z_i^{a_i}$, that is,
\begin{align*}
    P(z_1,\dots,z_k) = \sum_{\Vec{a}\in \mathcal{A}_k} \mathsf{coeff}(\Vec{a}) \cdot \prod_{i=1}^k z_i^{a_i}\,.
\end{align*}
By~\eqref{eq:poly_interpol_setup},  $P(z_1,\dots,z_k)=0$ for all $k$-tuples of positive integers $(z_1,\dots,z_k)$. By multivariate polynomial interpolation, only the constant $0$ polynomial can satisfy this conditions. Concretely, we obtain $\mathsf{coeff}(1,1,\dots,1)=0$. This   yields the desired  result since $\bij(X)$ is the set of functions $\tau\colon X \to V(H)$ with $d^\tau_i=1$ for all $i\in[k]$. Thus,
 $$
 \sum_{\tau \in \bij(X)} 
|\Ans^\tau((H,X),(\chi(F,\emptyset),c  ))| - 
|\Ans^\tau((H,X),(\chi(F,\{x_1\}),c))| \\
    = \mathsf{coeff}(1,1,\dots,1) = 0,$$
contradicting the assumption in the statement of the lemma, and
therefore
completing the proof.
\end{proof}

\begin{corollary}\label{cor:almost_final}
Suppose that the treewidth of~$F$ is $t$ and that $F$ is connected. 
Let $c = \gamma(\pi_1(\cdot))$.
Suppose that  
\[ \sum_{\tau \in \bij(X)} |\Ans^\tau((H,X),(\chi(F,\emptyset),
c ))| - |\Ans^\tau((H,X),(\chi(F,\{x_1\}),c))| \neq 0.\]
Then the WL-dimension of $(H,X)$ is at least $t$.
\end{corollary}
\begin{proof}
Let $\vv = (x_1,\ldots,x_k)$.
By Lemma~\ref{lem:main_cloning_lemma}, for every $k$-tuple of positive integers $\Vec{z}$,
$$
\newgraph(\chi(F,\emptyset),F,\pi_1,\vv,\vz) 
\cong_{(t-1)}  \newgraph(\chi(F,\{x_1\}),F,\pi_1,\vv,\vz).$$

From the definition of the cloning operation (Definition~\ref{def:cloning_operation}) and the fact that $\gamma$ is the identity on $X$ (see Definition~\ref{def:gamma}), 
 $\newgraph( \chi(F,\emptyset),F,\pi_1,\vv,\vz) = 
 \newgraph( \chi(F,\emptyset), F, c, \vv,\vz)$.
Similarly, 
$\newgraph( \chi(F,\{x_1\}),F,\pi_1,\vv,\vz) = 
 \newgraph( \chi(F,\{x_1\}), F, c, \vv,\vz)$.
So for every $k$-tuple of positive integers $\vz$,
$
\newgraph(\chi(F,\emptyset),F,c,\vv,\vz) 
\cong_{(t-1)}  \newgraph(\chi(F,\{x_1\}),F,c,\vv,\vz)$.

However, by Lemma~\ref{lem:isolate_colourful},  
there is a $k$-tuple $\vz=(z_1,\dots,z_k)$ of positive integers such that
\[ |\Ans((H,X),
\newgraph( \chi(F,\emptyset), F,c,\vv,\vz) )|  \neq |\Ans((H,X),
\newgraph( \chi(F,\{x_1\}), F,c,\vv,\vz) )|.\]
 
This shows that the function $G\mapsto |\Ans((H,X),G)|$ can distinguish $(t-1)$-WL invariant graphs and thus the WL-dimension of $(H,X)$ at least $t$.
\end{proof}

We will consider the set of partial automorphisms of a conjunctive query, defined as follows:
\begin{definition}[$\mathsf{Aut}(H,X)$]
$\mathsf{Aut}(H,X) :=\{\tau: X \to X \mid \exists a \in \mathsf{Aut}(H): a|_X = \tau\}$.
\end{definition}
Note that a mapping $\tau: X \to X$ can only be extended to an automorphism of $H$ if it is bijective. Thus:
\begin{observation}
  $\mathsf{Aut}(H,X)\subseteq \bij(X)$  
\end{observation}

The following result will be very useful for the remainder of this section; it can be found in Corollary 54.4 in the full version~\cite{DellRW19arxiv} of~\cite{DellRW19}.

\begin{lemma}[\cite{DellRW19arxiv,DellRW19}]\label{lem:minimal_auto}
Let $h$ be a homomosphism from~$H$ 
to~$H$ that 
surjectively maps~$X$ onto~$X$. 
If $(H,X)$ is counting minimal then $h$ is an automorphism of~$H$.
\end{lemma}

\begin{lemma}\label{lem:aut_magig_1}
    Let $G$ be a graph, let $c$ be a homomorphism from~$G$ to~$H$, and
let $\tau \in \bij(X) \setminus \mathsf{Aut}(H,X)$. If $(H,X)$ is counting minimal, then $|\Ans^\tau((H,X),(G,c))| = 0$.
\end{lemma}
\begin{proof}
Assume for contradiction that there is a homomorphism $h\in \Ans^\tau((H,X),(G,c))$. Then there is $\varphi \in \Hom(H,G)$ with $\varphi|_X = h$. Moreover,
$c(h(\cdot))=\tau(\cdot)$. Thus $c(\varphi(\cdot))$ is a
homomorphism from~$H$ to~$H$ such that,
for all $x\in X$, $c(\varphi(x))=\tau(x)$. Since $\tau\in \bij(X)$, $c(\varphi(\cdot))$ maps $X$ surjectively to itself. By Lemma~\ref{lem:minimal_auto}, $c(\varphi(\cdot))$ is an automorphism of~$H$, and thus $\tau \in \mathsf{Aut}(H,X)$,  contradicting the assumption in the statement of the lemma that $\tau \notin  \mathsf{Aut}(H,X)$.
\end{proof}

\begin{lemma}\label{lem:aut_magic_2}
    Let $G$ be a graph, let $c$ be a homomorphism from~$G$ to~$H$, and
let $\tau \in \mathsf{Aut}(H,X)$. Then
$|\Ans^\tau((H,X),(G,c))| = |\Ans^\mathsf{id}((H,X),(G,c))|$.
\end{lemma}
\begin{proof}
We will construct a bijection between  
$\Ans^\tau((H,X),(G,c))$ and $\Ans^\mathsf{id}((H,X),(G,c))$.

Since $\tau \in \mathsf{Aut}(H,X)$ there is an automorphism $a$ of $H$ such that $a|_X=\tau$. Let $a^{-1}$ be the inverse of $a$ 
in the group   $\mathsf{Aut}(H)$ and observe that $\tau^{-1} =a^{-1}|_X$.

We first show that, for every function $h\in \Ans^\tau((H,X),(G,c))$,
the function $h(\tau^{-1}(\cdot))$ is in 
$\Ans^\mathsf{id}((H,X),(G,c))$.
To see this, consider 
$h\in \Ans^\tau((H,X),(G,c))$.
From the definition of $\Ans^\tau((H,X),(G,c))$,
there is  a homomorphism $\varphi \in \Hom(H,G)$
with $\varphi|_X=h$ and that $c(h(\cdot))=\tau$. 
Since $a^{-1}$ is an automorphism of $H$,  $\varphi(a^{-1}(\cdot))\in \Hom(H,G)$. Also, for every $x\in X$,
$\varphi(a^{-1}(x))=\varphi(\tau^{-1}(x))$.
Since $\tau^{-1}(x)\in X$ and $\varphi|_X = h$,
$\varphi(\tau^{-1}(x))=h(\tau^{-1}(x))$.
Putting these equalities together, 
$\varphi(a^{-1}(x))=h(\tau^{-1}(x))$.
Thus $h(\tau^{-1}(\cdot))\in \Ans((H,X),G)$. 
Finally, since $h\in \Ans^\tau((H,X),(G,c))$,
$c(h(\tau^{-1}(\cdot))) = \tau(\tau^{-1}(\cdot))=\mathsf{id}(\cdot)$ so $h(\tau^{-1}(\cdot))\in \Ans^\mathsf{id}((H,X),(G,c))$, as required.

Let $b$ be the function that maps every  
$h\in \Ans^\tau((H,X),(G,c))$ to the
function $b(h) := h(\tau^{-1}(\cdot))$.
We have shown that $b$ is a map from 
$\Ans^\tau((H,X),(G,c))$ to $\Ans^\mathsf{id}((H,X),(G,c))$.
By a symmetric argument, the function~$\hat{b}$
that maps every $h\in \Ans^\mathsf{id}((H,X),(G,c))$
to $\hat{b}(h) := h(\tau(\cdot))$
is a map from $\Ans^\mathsf{id}((H,X),(G,c))$ to 
$\Ans^\tau((H,X),(G,c))$.

Since $b(\hat{b}(\cdot))=\mathsf{id}$ and $\hat{b}(b(\cdot))=\mathsf{id}$, $b$ and $\hat{b}$ are bijections, completing the proof. 
\end{proof}

Corollary~\ref{cor:tool} is the main tool that 
we will use to lower bound the WL-dimension of conjunctive queries.

\begin{corollary}\label{cor:tool}
Suppose that
$\mathsf{tw}(F)= t>1$, that $F$ is connected, and that  $(H,X)$ is  counting minimal.
Let $c = \gamma(\pi_1(\cdot))$.
Suppose that
$ |\Ans^{\mathsf{id}}((H,X),(\chi(F,\emptyset),c ))|$ is
not equal to $|\Ans^{\mathsf{id}}((H,X),(\chi(F,\{x_1\}),c ))|$. Then the WL-dimension of $(H,X)$ is at least $t$.
\end{corollary}
\begin{proof}
By Corollary~\ref{cor:almost_final} it suffices to establish the following.
\[ \sum_{\tau \in \bij(X)} |\Ans^\tau((H,X),(\chi(F,\emptyset),
c ))| - |\Ans^\tau((H,X),(\chi(F,\{x_1\}),c))| \neq 0.\]
By Observation~\ref{obs:gamma_pi1_Hcol}, $c$ is an $H$-colouring of both $\chi(F,\emptyset)$ and $\chi(F,\{x_1\})$. Thus we can use Lemmas~\ref{lem:aut_magig_1} and~\ref{lem:aut_magic_2}, and obtain:
\begin{align*}
  ~&~  \sum_{\tau \in \bij(X)} |\Ans^\tau((H,X),(\chi(F,\emptyset),
c ))| - |\Ans^\tau((H,X),(\chi(F,\{x_1\}),c))|\\
=&~ \sum_{\tau \in \mathsf{Aut}(H,X)} |\Ans^\tau((H,X),(\chi(F,\emptyset),
c ))| - |\Ans^\tau((H,X),(\chi(F,\{x_1\}),c))|\\
=&~ |\mathsf{Aut}(H,X)| \cdot \left( |\Ans^{\mathsf{id}}((H,X),(\chi(F,\emptyset),c ))| - |\Ans^{\mathsf{id}}((H,X),(\chi(F,\{x_1\}),c ))|\right).
\end{align*}
The first factor is non-zero because $\mathsf{Aut}(H,X)$ is non-empty. The second factor is non-zero by the assumption in the statement of the corollary. 
\end{proof}

\subsection{Proving the Lower Bound}

 We first set up some notation, following~\cite{DellRW19}.

\begin{definition}[Colour-prescribed Homomorphisms]
Let $H$ and $G$ be graphs, let $(H,X)$ be a conjunctive query, and let $c$ be an $H$-colouring of $G$. 
Define
\begin{align*} \cpHom(H,(G,c)) &= 
\{ h\in \Hom(H,G) \mid  \forall v \in V(H),
c(h(v))=v\}, \mbox{and} \\
\cpAns((H,X),(G,c)) &= 
\{a \colon X \to V(G) \mid 
\mbox{
$\exists h\in \cpHom(H,(G,c))$ such that $h|_X=a$ }  \}.
\end{align*}
Homomorphisms in $\cpHom(H,(G,c))$ are said to be
``colour-prescribed'' with respect to~$c$.
\end{definition}

\begin{observation}\label{obs:oneway}
$\cpAns((H,X),(G,c)) \subseteq \Ans^{\mathsf{id}}((H,X),(G,c ))$   
\end{observation}

Observation~\ref{obs:oneway} follows directly
from the definitions.
Recall that
$\Ans^{\mathsf{id}}((H,X),(G,c)) = \{h \in \Ans((H,X),G)\mid 
\forall v\in X,
c(h(v)) = v  \} \}$.
So $c(h(v))=v$ is required for all $v\in X$
in the definition of 
$\Ans^{\mathsf{id}}((H,X),(G,c))$, while for all $v\in V(H)$ in the definition of $\cpAns((H,X),(G,c))$.
Lemma~\ref{lem:setsequal} shows equivalence for counting minimal conjunctive queries.

\begin{lemma}\label{lem:setsequal}
Let $(H,X)$ be a counting minimal conjunctive query, let $G$ be a graph, and let~$c$ be an $H$-colouring of $G$. Then $\cpAns((H,X),(G,c)) = \Ans^{\mathsf{id}}((H,X),(G,c ))$.
\end{lemma}

\begin{proof}
Observation~\ref{obs:oneway} proves one direction. For the other direction, consider a map $a \in \Ans^{\mathsf{id}}((H,X),(G,c ))$. From the definition, there is a homomorphism $h \in \Hom(H,G)$ such that $h|_X = a$ and $c(a(x))=x$ for all $x \in X$. From the definitions of~$c$ and~$h$, the function $\varphi:=c(h(\cdot))$ is a homomorphism from~$H$ to itself. Since $\varphi$ maps $X$ to $X$
and  $(H,X)$ is counting minimal,  Lemma~\ref{lem:minimal_auto} guarantees that $\varphi$ is an automorphism of~$H$. Let $h'=h(\varphi^{-1}(\cdot))$. Clearly, $h'\in \Hom(H,G)$. Moreover, for each $v\in V(H)$, $c(h'(v)) = c(h(\varphi^{-1}(v)))=\varphi(\varphi^{-1})(v) = v$.
Thus $h'\in \cpHom(H,(G,c))$. Since $\varphi$ is the identity on $X$, the same is true of~$\varphi^{-1}$. Thus,   for all $x\in X$, $h'(x)= h(\varphi^{-1}(x))=h(x)=a(x)$, implying that $a \in \cpAns((H,X),(G,c))$.  
\end{proof}

The following definition will be useful; it will provide a parity condition which is both sufficient and necessary for containment in the relevant set of answers.

\begin{definition}[Extendable Assignments, $\mathcal{E}(X,F,W)$]\label{def:extendables}\label{def:extendable}
Let $(H,X)$ be a conjunctive query with $ \emptyset \subsetneq X=\{x_1,\ldots,x_k\} \subsetneq V(H) $ and suppose that $x_1$ is adjacent to at least one vertex in $Y=V(H)\setminus X$. Suppose that $H$ is connected and let $\ell$ be an odd positive integer.  Let $F=F_\ell(H,X)$ and let $c=\gamma(\pi_1(\cdot))$.
Let $C_1,\dots,C_m$ be the connected components of $H[Y]$.
For each $i\in[m]$, denote the vertex sets of the $\ell$ copies of $C_i$ in $F$ by $V_i^1,\dots,V_i^\ell$.
Let $W$ be a subset of $X$ and let
$\varphi$ be an assignment from $X$ to $V(\chi(F,W))$ 
such that, for all $p\in [k]$, $c(\varphi(x_p)) = x_p$.
Define $S_1,\ldots,S_k$ such that  $\varphi(x_p)=(x_p,S_p)$ for all $p\in[k]$. 
We say that $\varphi$ is \emph{extendable} if the following two conditions hold.
\begin{itemize}
\item[(E1)] For every $\{x_a,x_b\} \in E(H[X])$, $x_a \in S_b \Longleftrightarrow x_b \in S_a$.
\item[(E2)] For every $i\in [m]$ there is a $j\in [\ell]$ such that
$ \sum_{p=1}^k |S_p \cap V_i^j| $ is even.   
\end{itemize}
Define 
$\mathcal{E}(X,F,W) := \{ \varphi : X \to \chi(F,W) \mid 
\mbox{ $\varphi$ is extendable and $\forall x\in X$, $c(\phi(x))=x$} \}$.  
\end{definition}

We will prove  
Lemmas~\ref{lem:distinguish_extendables},
\ref{lem:former_claim1} and \ref{lem:former_claim2}
in Section~\ref{sec:proofs}.

\begin{lemma}\label{lem:distinguish_extendables}
 Let $(H,X)$ be a conjunctive query with $ \emptyset \subsetneq X \subsetneq V(H) $ and suppose that $x_1$ is adjacent to at least one vertex in $V(H)\setminus X$. Suppose that $H$ is connected and let $\ell$ be an odd positive integer.  Let $F=F_\ell(H,X)$ and let $c=\gamma(\pi_1(\cdot))$. 
Then  $|\mathcal{E}(X,F,\emptyset)|> |\mathcal{E}(X,F,\{x_1\})|$.
\end{lemma}
\begin{lemma}\label{lem:former_claim1}
Let $(H,X)$ be a conjunctive query with $ \emptyset \subsetneq X=\{x_1,\ldots,x_k\} \subsetneq V(H) $ and suppose that $x_1$ is adjacent to at least one vertex in $V(H)\setminus X$. Suppose that $H$ is connected and let $\ell$ be an odd positive integer.  Let $F=F_\ell(H,X)$ and let $c=\gamma(\pi_1(\cdot))$. 
Let $W$ be a subset of $X$ and let
$\varphi$ be an assignment from $X$ to $V(\chi(F,W))$ 
such that, for all $p\in [k]$, $c(\varphi(x_p)) = x_p$.
If $\varphi$ is not extendable, then $\varphi \notin \cpAns((H,X),(\chi(F,W),c ))$.
\end{lemma}

\begin{lemma}\label{lem:former_claim2}
Let $(H,X)$ be a conjunctive query with $ \emptyset \subsetneq X=\{x_1,\ldots,x_k\} \subsetneq V(H) $ and suppose that $x_1$ is adjacent to at least one vertex in $V(H)\setminus X$. Suppose that $H$ is connected and let $\ell$ be an odd positive integer.  Let $F=F_\ell(H,X)$ and let $c=\gamma(\pi_1(\cdot))$. 
Let $W$ be a subset of $X$ and let
$\varphi$ be an assignment from $X$ to $V(\chi(F,W))$ 
such that, for all $p\in [k]$, $c(\varphi(x_p)) = x_p$.
If $\varphi$ is extendable, then $\varphi \in \cpAns((H,X),(\chi(F,W),c ))$
\end{lemma}

\begin{lemma}\label{lem:paritysame}
Let $(H,X)$ be a conjunctive query with $ \emptyset \subsetneq X \subsetneq V(H) $ and suppose that $x_1$ is adjacent to at least one vertex in $V(H)\setminus X$. Suppose that $H$ is connected and let $\ell$ be an odd positive integer.  Let $F=F_\ell(H,X)$ and let $c=\gamma(\pi_1(\cdot))$. Let $W$ be a subset of $X$. Then 
$\cpAns((H,X),(\chi(F,W),c )) = \calE(X,F,W)$.
\end{lemma}

\begin{proof} 
Let $\phi$ be an assignment from $X$ to $V(\chi(F,W))$.
If there is an $x\in X$ such that $c(\phi(x)) \neq x$ then
$\phi \notin \calE(X,F,W)$ and
$\phi \notin \cpAns((H,X),(\chi(F,W),c ))$.
Otherwise, if $\phi$ is not extendable then it not in 
$\calE(X,F,W)$ by definition and it is not in 
$\cpAns((H,X),(\chi(F,W),c ))$ by 
Lemma~\ref{lem:former_claim1}.
If $\phi$ is extendable then it is in 
$\calE(X,F,W)$ by definition, Also, by 
Lemma~\ref{lem:former_claim2}, $\phi$
is  in 
$\cpAns((H,X),(\chi(F,W),c ))$. 
\end{proof}

\begin{lemma}\label{lem:main_technical_WL_LB}
Let $(H,X)$ be a conjunctive query with $ \emptyset \subsetneq X \subsetneq V(H) $ and suppose that $x_1$ is adjacent to at least one vertex in $V(H)\setminus X$. Suppose that $H$ is connected and let $\ell$ be an odd positive integer.  Let $F=F_\ell(H,X)$ and let $c=\gamma(\pi_1(\cdot))$. Then
$ |\cpAns((H,X),(\chi(F,\emptyset),c ))| > |\cpAns((H,X),(\chi(F,\{x_1\}),c ))|$.
\end{lemma}

\begin{proof} 
The lemma follows immediately from 
Lemmas~\ref{lem:distinguish_extendables} and~\ref{lem:paritysame}.
\end{proof}

We can now state and prove Lemma~\ref{lem:thisgoal},
the main goal of this subsection, which will enable us to immediately infer Theorem~\ref{thm:lower_bound}.

\begin{lemma}\label{lem:thisgoal}
Let $(H,X)$ be a counting minimal conjunctive query 
such that $H$ is connected, and 
$ \emptyset \subsetneq X \subsetneq V(H) $.
Without loss of generality, suppose
that $x_1$ is adjacent to at least one vertex in $V(H)\setminus X$. Let $\ell$ be an odd positive integer. 
Let $F=F_\ell(H,X)$ and let $c=\gamma(\pi_1(\cdot))$. 
Then
$$
|\Ans^{\mathsf{id}}((H,X),(\chi(F,\emptyset),c ))| > |\Ans^{\mathsf{id}}((H,X),(\chi(F,\{x_1\}),c ))|
$$
\end{lemma}
\begin{proof} 
The lemma follows directly from Lemma~\ref{lem:setsequal}
and lemma~\ref{lem:main_technical_WL_LB}.
\end{proof}

We can now prove Theorem~\ref{thm:lower_bound}.
\thmlowerbound*

\begin{proof} 
By Corollary~\ref{cor:ewidth_alternative}, $\ctw(H,X)= \max\{\tw(F_\ell(H,X))\mid \ell>0 \}$. Choose $\ell$ large enough such that $\ctw(H,X)=\tw(F_\ell(H,X))$. We can assume, without loss of generality, that $\ell$ is odd: If $\ell$ is even, choose $\ell+1$ instead and note that $\tw(F_{\ell+1}(H,X))\geq \tw(F_\ell(H,X))$ since $F_\ell(H,X)$ is a subgraph of $F_{\ell+1}(H,X)$ and since treewidth is monotone under taking subgraphs.

Let $F=F_\ell(H,X)$ (note that $F$ is connected as $(H,X)$ is connected) and let $c = \gamma(\pi_1(\cdot))$. By Corollary~\ref{cor:tool}, it suffices to show that
$ |\Ans^{\mathsf{id}}((H,X),(\chi(F,\emptyset),c ))|$ is
not equal to $|\Ans^{\mathsf{id}}((H,X),(\chi(F,\{x_1\}),c ))|$, which holds by Lemma~\ref{lem:thisgoal}. This concludes the proof.
\end{proof}

\subsubsection{Proofs of Lemmas~\ref{lem:distinguish_extendables},
\ref{lem:former_claim1} and \ref{lem:former_claim2}}\label{sec:proofs}

This section will use the following 
combinatorial lemma,  Lemma~\ref{lem:connected_parity},
which seems to be folklore.

\begin{lemma}\label{lem:connected_parity}
Let $G$ be a connected graph and 
let $S\subseteq V(G)$ be a vertex-subset of even cardinality. Then there is an assignment $\beta: E(G) \to \{0,1\}$ such that, for all $v\in V(G)$, 
\begin{equation}\label{eq:comb_lemma}
\sum_{u\in N(v)} \beta(\{u,v\}) \equiv \begin{cases} 1 \mod 2 & \mbox{if $v \in S$,} \\ 0 \mod 2 & \mbox{if $v \notin S$.} \end{cases}
\end{equation}
\end{lemma}
\begin{proof}
The proof is by induction on $n=|V(G)|$. If $n=1$, then $E(G)=\emptyset$    so necessarily $S=\emptyset$. 
By convention, the (empty) sum is~$0$ which is as desired since $v\notin S$.

Assume for the induction hypothesis that the claim holds for $n$ and let $G=(V,E)$ be a graph with $n+1$ vertices. 
Let $v$ be a vertex of~$G$ such that $G' := G\setminus \{v\}$ is connected.  Let $v_1,\dots,v_d$ be the neighbours of $v$. Note that $d>0$ since $G$ is connected. We consider two cases.

{\bf Case 1. If $v\in S$:}  Set $\beta^\ast(\{v,v_1\})=1$ and $\beta^\ast(\{v,v_i\})=0$ for all $1<i\leq d$. Set $S'= (S\setminus v) \oplus v_1$. Note that $S'$ has even cardinality, since $S \setminus v$ has odd cardinality. We can thus apply the induction hypothesis to $G'$ and $S'$ and obtain a function $\beta':E(G') \to \{0,1\}$ that satisfies~(\ref{eq:comb_lemma}) for $G'$. It is easy to see that $\beta^\ast\cup\beta'$ satisfies (\ref{eq:comb_lemma}) for $G$.
   
{\bf Case 2. If $v\notin S$:}   Set $\beta^\ast(\{v,v_i\})=0$ for all $i\in[d]$.  
We apply the induction hypothesis to $G'$ and $S$ and obtain a function $\beta':E(G') \to \{0,1\}$ that satisfies~(\ref{eq:comb_lemma}) in $G'$. It is easy to see that $\beta^\ast\cup\beta'$ satisfies (\ref{eq:comb_lemma}) in $G$.
 
\end{proof}

We finish by proving Lemmas~\ref{lem:former_claim1},
\ref{lem:former_claim2} and~\ref{lem:distinguish_extendables}.

\begin{proof}[Proof of Lemma~\ref{lem:former_claim1}]
Let $Y = V(H) \setminus X$. Let $G=\chi(F,W)$.
For all $p\in [k]$, $i\in [m] $
and $j\in [\ell]$, define $S_p$, $C_i$, and $V_i^j$  as in Definition~\ref{def:extendable}.
If (E1) is not satisfied, then $\varphi$ cannot be extended to a homomorphism from~$H$ to~$G$.
If (E2) is not satisfied then there is an $i\in[m]$ such that, for all $j\in[\ell]$, 
$\sum_{p=1}^k |S_p \cap V_i^j|$ is odd. 
Fix this~$i$.
Assume for contradiction that $\varphi \in \cpAns((H,X),(G,c ))$. Then there is a   homomorphism $h\in \cpHom(H,(G,c))$ with $h|_X=\varphi$.

Let $y_1,\dots,y_t$ be the vertices of~$C_i$. Since $h\in \cpHom(H,(G,c))$, the definition of~$c$ 
implies that for all $v\in V(H)$, $\gamma(\pi_1(h(v)))=v$.
Recall that the vertices of 
$F=F_\ell(H,X)$ are $X \cup (Y \times [\ell])$ and that the vertices of
$G=\chi(F,W)$
are pairs $(w,T)$ where $w$ is a vertex of~$F$
and $T\subseteq N(w)$.  
So for a vertex $v=y_s$ of~$C_i$,
$\gamma(\pi_1(h(y_s)))=y_s$ implies that $h(y_s)$
is of the form $((y_s,j),T_s)$ for some $j\in [\ell]$.
Moreover, since $C_i$ is a connected component of~$H[Y]$
there is a single $j$ such that, for all $s\in [t]$,
$h(y_s) =  ((y_s,j),T_s)$. 

We next define a $t$-by-$(k+t)$ matrix $M$, indexed by vertices of~$F$. The rows of $M$ are indexed by $(y_1,j),\dots,(y_t,j)$, and the columns of $M$ are indexed $x_1,\dots,x_k,(y_1,j),\dots,(y_t,j)$.   The entries of $M$ are defined as follows. For $s\in [t]$,
\[ M((y_s,j),v)=\begin{cases} 1 & \mbox{if $v\in T_s$,} \\ 0 & \mbox{if $v \notin T_s$}. \end{cases} \]
From the definition of $G=\chi(F,W)$,
the size of $T_s$ is odd if $(y_s,j) \in W$ and even otherwise. Since $W$ is a subset of $X$, 
every set $T_s$ has even cardinality. 
Thus every row of~$M$ has an even number of $1$s, and therefore $M$ has an even number of $1$s. Let $M_X$ be the submatrix of $M$ containing only the columns indexed by $x_1,\dots,x_k$. Let $M_Y$ be the submatrix of $M$ containing only the columns indexed by  $(y_1,j),\dots,(y_t,j)$.

Note that $M_Y$ is a square matrix. We next show that it is symmetric.  If $M((y_s,j),(y_{s'},j))=1$, then $(y_{s'},j) \in T_s$, which implies that $(y_s,j)$ and $(y_{s'},j)$ are adjacent in $F$ and thus $y_s$ and $y_{s'}$ are adjacent in $H$. Suppose for contradiction that $M((y_{s'},j),(y_s,j))=0$. Then $(y_{s},j) \notin T_{s'}$. So $((y_s,j),T_s)$ and $((y_{s'},j),T_{s'})$ are not adjacent in $G=\chi(F, W)$, contradicting the fact that $h$ is a homomorphism from~$H$ to~$G$. Thus, our assumption was wrong, and $M_Y$ is symmetric. Since the diagonal of $M_Y$ contains only $0$s (since $H$ and $F$ do not have self-loops), the number of $1$s in $M_Y$ is even.

To finish the proof 
we will show that, for every $p \in [k]$,
the number of $1$s in the column of~$M$ indexed by~$x_p$ is  $|S_p \cap V_i^j|$ so the number of $1$s in $M_X$ 
is $\sum_{p=1}^k |S_p \cap V_i^j|$, which is odd by the choice of~$i$. Thus, $M_X$ has an odd number of $1$s and $M_Y$ has an even number of $1$s, contradicting the fact that $M$ has an even number of $1$s. We conclude that
our initial assumption for contradiction, that 
$\varphi \in \cpAns((H,X),(G,c ))$, is false, proving the lemma.

So to finish, fix $p \in [k]$.
We will show that the number of $1$s in the column of~$M$ indexed by~$x_p$ is $|S_p \cap V_i^j|$.

Before considering the entries of this column of~$M$, we establish a fact which we will use twice ---
if $x_p$ and $(y_s,j)$ are adjacent in $F$ then, from the definition of~$F$,
$x_p$ and $y_s$ are adjacent in~$H$.
Since $h$ is a homomorphism from~$H$ to~$G$ that extends $\varphi$, from the definition of~$S_p$, 
$h(x_p) = \phi(x_p) = (x_p,S_p)$, so
it follows that 
$(x_p,S_p)$ is adjacent to $h(y_s) = ((y_s,j),T_s)$ in~$G$. 
Using this fact, we will show that the number of $1$s in the column of~$M$ indexed by~$x_k$ is $|S_p \cap V_i^j|$.

First, consider any $s$ such that that $M((y_s,j),x_p)=1$.  
From the definition of~$M$,
$x_p\in T_s$ so, 
from the definition of $G=\chi(F,W)$,
$x_p$ and $(y_s,j)$ are adjacent in $F$.
From the fact,  
$(x_p,S_p)$ is adjacent to $((y_s,j),T_s)$ in~$G$. From the definition of $G$, $(y_s,j) \in S_p$.
By construction, $(y_s,j) \in V_i^j$, so 
$(y_s,j) \in S_p \cap V_i^j$.

Finally, consider any $s$ such that that $M((y_s,j),x_p)=0$. From the definition of $M$,
$x_p \notin T_s$. There are two cases.

Case 1:
If  $x_p$ and $(y_s,j)$ are not adjacent in $F$ then,
from the definition of $S_p$, $(y_s,j)$ is not in $S_p$ so it is clearly not in $S_p \cap V_i^j$.

Case 2. If $x_p$ and $(y_s,j)$ are adjacent in $F$
but $x_p \notin T_s$  then from the fact 
$(x_p,S_p)$ is adjacent to $((y_s,j),T_s)$ in~$G$. 
So from the definition of $G$, $(y_s,j) \notin S_p$
so it is clearly not in $S_p \cap V_i^j$.
 
As required, we thus obtain that 
the number of $1$s in the column of~$M$ indexed by~$x_k$ is $|S_p \cap V_i^j|$.

\end{proof}

\begin{proof}[Proof of Lemma~\ref{lem:former_claim2}]
Let $Y = V(H) \setminus X$. Let $G=\chi(F,W)$.
For all $p\in [k]$, $i\in [m] $
and $j\in [\ell]$, define $S_p$, $C_i$, and $V_i^j$  as in Definition~\ref{def:extendable}.
Condition (E1) ensures that, for every $\{x_a,x_b\} \in E(H[X])$, $(x_a,S_a)$ and $(x_b,S_b)$ are adjacent in~$G$.
By (E2), for every $i\in [m]$ there is a $j_i \in [\ell]$ such that $\sum_{p=1}^k |S_p \cap V_i^{j_i}|$ is even. 

To show that $\phi \in \cpAns((H,X),(G,c))$,
we will construct an $h\in \cpHom(H,(G,c))$ with $h|_X = \varphi$. 
To do this we will choose a value for $h(y)$ for each $y\in Y$. 
For every $i\in[m]$, let 
$V(C_i) := \{y_{i,1},\dots,y_{i,t_i}\}$ be the vertices of the connected component $C_i$ of $H[Y]$. 
Given a vertex $y_{i,s} \in V(C_i)$,
let $N_i(y_{i,s})$ denote the set of its neighbours in~$C_i$.
For every $y_{i,s} \in V(C_i)$,
let $T_{i,s,X} = \{ x_p \in X \mid (y_{i,s},j_i) \in S_p\}$.

Let $\Omega_i = \{ y_{i,s} \>:\>
\mbox{$|T_{i,s,X}|$ is odd}\}$.
We will show that $|\Omega_i|$ is even.
To see this, let $M_X$ be a matrix with rows indexed by
$(y_{i,1},j_i),\ldots,(y_{i,t_i},j_i)$ and columns indexed by~$X$. Define $M( y_{i,s},j_i),x_p)$ to be~$1$ if 
$x_p \in T_{i,s,X}$ and~$0$ otherwise.
By construction, 
the number of $1$s in the column of $M$ indexed by~$x_p$ is the number of vertices $(y_{i,s},j_i) \in S_p$ 
so this is 
$|S_p \cap V_i^{j_i}|$.
By the choice of~$j_i$, 
the number of $1$s in $M_X$ is even.
So $\Omega_i$, which contains the indices of rows with an odd number of $1$s, has even cardinality.
Apply Lemma~\ref{lem:connected_parity} with graph
$C_i$ and $S=\Omega_i$ to obtain an assignment 
$\beta \colon E(C_i) \to \{0,1\}$ such that,
for all $y_{i,s}\in \Omega_i$, 
$ \sum_{y_{i,s'}\in N_i(y_{i,s})} \beta(\{y_{i,s'},y_{i,s}\})$ is odd and
for all $y_{i,s} \in V(C_i)\setminus \Omega_i$,
this sum is even.
Finally, let $T_{i,s,Y} = 
\{ (y_{i,s'},j_i)   
\mid \mbox{$y_{i,s'} \in N_i(y_{i,s}) $ and 
$\beta(\{y_{i,s'},y_{i,s}\})=1$}\}
$ and
let $T_{i,s} = T_{i,s,X} \cup T_{i,s,Y}$. 
We then define 
$h(y_{i,s}) = ((y_{i,s},j_i),T_{i,s})$. 
Our goal is to show that $h\in \cpHom(H,(G,c))$.
For this we require
\begin{itemize}
\item Property 1: Each $h(y_{i,s})$ is a vertex of~$G$.
\item Property 2: For all $y_{i,s}\in Y$, $c(h(y_{i,s}))=y_{i,s}$. (This follows immediately from the definition of~$h$.)
\item Property 3: $h$ is a homomorphism from~$H$ to~$G$. 
That is, all three of the following hold.
\begin{itemize}
\item Property 3a: For every edge $\{x_a,x_b\}$ of $E(H[X])$, $\{h(x_a),h(x_b)\}$ is an edge of $G$ (this follows immediately from (E1) and the definition of~$G$).
\item Property 3b: For every edge $\{x_p, y_{i,s}\}$ of~$H$ with $x_p\in X$ and $y_{i,s} \in C_i$, 
$\{h(x_p),h(y_{i,s})\}$ is an edge of $G$.
\item Property 3c: For every edge $\{ y_{i,s},y_{i,s'}\}$ of $C_i$,
$\{ h(y_{i,s}), h(y_{i,s'})\}$ is an edge of~$G$.
\end{itemize} 
\end{itemize}
We start by showing Property~1 --- that each $h(y_{i,s})$ is a vertex of~$G$. For this, we need two constraints to be satisfied as follows.
\begin{itemize} 
\item Constraint 1: Every vertex of~$T_{i,s}$  
must be a neighbour of~$(y_{i,s},j_i)$ in~$F$.
\item 
Constraint 2: $|T_{i,s}|$ must be even (this constraint comes from the definition of~$G$
since $y_{i,s}$ is not in $W$, which is a subset of~$X$).
\end{itemize}
Constraint~1 follows directly from the definition of~$T_{i,s}$.
To see that every $x_p \in T_{i,s,X}$ is a neighbour of  $(y_{i,s},j_i) $ note that  $(y_{i,s},j_i) \in S_p$ which implies 
that $\{x_p,(y_{i,s},j_i)\}$ is an edge of~$F$ since
$(x_p,S_p)$ is a vertex of~$G$. It is immediate from the definition of $T_{i,s,Y}$ and $F$ that every vertex in $T_{i,s,Y}$ 
is a neighbour of  $(y_{i,s},j_i) $ in~$F$.

Constraint~2 is by construction. Consider any pair $(i,s)$. If $|T_{i,s,X}|$ is even then $y_{i,s} \notin \Omega_i$ so $|T_{i,s,Y}|$ is even. On the other hand, if $|T_{i,s,X}|$ is odd then $Y_{i,s} \in \Omega_i$, so $|T_{i,s,Y}|$ is odd.

We next consider Property 3b. Consider an edge $\{x_p, y_{i,s}\}$ of~$H$ with $x_p\in X$ and $y_{i,s} \in C_i$.
Then $h(x_p) = (x_p,S_p)$. If $(y_{i,s},j_i) \in S_p$ then $x_p \in T_{i,s,X}$ so $x_p \in T_{i,s}$ 
and $h(y_{i,s}) = ((y_{i,s},j_i),T_{i,s})$ is connected to $(x_p,S_p)$ in $G$ by the definitions of~$F$ and~$G$.
Similarly, if $(y_{i,s},j_i)\notin S_p$ then $x_p \notin T_{i,s}$ so again 
$((y_{i,s},j_i),T_{i,s})$ is connected to $(x_p,S_p)$ in $G$.

To finish the proof, we establish Property~3c.
Consider an edge $\{ y_{i,s},y_{i,s'}\}$ of $C_i$.
Then $h(y_{i,s}) = ((y_{i,s},j_i),T_{i,s})$ and $h(y_{i,s'}) = ((y_{i,s'},j_i),T_{i,s'})$
By construction, $(y_{i,s'},j_i) \in T_{i,s}$ iff $(y_{i,s},j_i) \in T_{i,s'}$. Hence $h(y_{i,s})$ and
$h(y_{i,s'})$ are connected in~$G$.
\end{proof}

\begin{proof}[Proof of Lemma~\ref{lem:distinguish_extendables}] 
Let $Y = V(H) \setminus X$. For $i\in [m]$ and $j\in [\ell]$, define $C_i$ and $V_i^j$  as in Definition~\ref{def:extendable}.
For every  
subset $W$ of $X$ and every
map $\phi \colon X \to V(\chi(F,W))$ that satisfies
$c(\phi(x_p)) = x_p$ for all $p\in [k]$, let $S_1(\phi),\ldots,S_k(\phi)$ be the sets 
defined in Definition~\ref{def:extendable}. 
For every $\phi \in \calE(X,F,W)$, (E2) guarantees that for every $i\in [m]$ there is a $j_i\in [\ell]$
such that $\sum_{p=1}^k |S_p(\phi) \cap V_i^{j_i}|$ is even.
We will partition $\mathcal{E}(X,F,W)$ into disjoint sets in terms of~$i$ and~$j_i$ as follows.
First we define $\mathcal{E}(X,F,W,1)$ by fixing $i=1$ and $j_i>1$.
 $$
\mathcal{E}(X,F,W,1) := \{ \varphi \in \mathcal{E}(X,F,W) \mid \exists j_1>1: 
\mbox{$\sum_{p=1}^k |S_p(\phi) \cap V_1^{j_1}|$ is even} \}.
 $$
For all $i\in \{2,\ldots,m\}$, define
$$
\mathcal{E}(X,F,W,i) := \{ \varphi \in \mathcal{E}(X,F,W)\setminus 
(\cup_{q=1}^{i-1}
\mathcal{E}(X,F,W,q)) \mid \exists j_i>1: 
\mbox{$\sum_{p=1}^k |S_p(\phi) \cap V_i^{j_i}|$ is even}  \}.
$$
Finally, define
$
\mathcal{E}(X,F,W,0) :=  \mathcal{E}(X,F,W) \setminus 
(\bigcup_{i=1}^m \mathcal{E}(X,F,W,i) )
$.

Since the sets $\calE(X,F,W,0),\ldots,\calE(X,F,W,m)$ are a disjoint partition of 
$\calE(X,F,W)$ for every subset $W$ of~$X$, the lemma follows immediately from the following three claims, which clearly imply $|\calE(X,F,\emptyset)| > |\calE(X,F,\{x_1\}|$, as required.

\noindent{\bf Claim 1.} For all $i\in [m]$, 
$|\calE(X,F,\emptyset,i)| = |\calE(X,F,\{x_1\},i)|$.

Fix $i\in [m]$.
To prove Claim~1, we construct a bijection~$b$ from~$\calE(X,F,\emptyset,i)$ to~$\calE(X,F,\{x_1\},i)$. 
Since $H$ is connected, there is path from $C_i$ to~$x_1$
in $H$ that starts at some vertex $y\in C_i$,
takes an edge from~$y$ 
to~$X$ and does not re-visit $C_i$ before reaching $x_1$. 
So there is a path from $(y,1)$ to~$x_1$ in~$F$
whose vertices are  in $X \cup \{ (y',1) \mid y' \in Y\}$
and is of the form 
$P = (y,1) x_{t_1} P_1 x_{t_2} P_2 x_{t_3} P_3 \dots x_{t_{s-1}} P_{s-1} x_{t_{s}}$  
where  
$x_{t_1},\ldots,x_{t_{s}}$ are distinct vertices in~$X$ 
with $t_{s}=1$
and
each $P_j$ is either empty or 
for some $i(j) \in [m]\setminus \{i\}$,
it is 
a non-empty simple path in~$F$ whose vertices are in $V_{i(j)}^1$.
Without loss of generality, 
the path visits every connected component
of $H[Y]$ at most once so for any distinct $j$ and $j'$
such that $P_j$ and $P_{j'}$ are both non-empty, 
$i(j) \neq i(j')$.

Our goal is to define the bijection~$b$.
For each $\phi\in\calE(X,F,\emptyset,i)$
and each $p\in [k]$,
we first define a subset $S'_p(\phi)$ of vertices of~$F$.
Let $N_P(x_p)$ be the neighbours of~$x_p$ in the path~$P$
(if $x_p$ is not in the path~$P$, then $N_P(x_p)=\emptyset$).
Then define $S'_p(\phi) = S_p(\phi) \oplus N_P(x_p)$.
Finally, we define   $b(\phi)$ to be the map
which maps every $p\in [k]$ to $(x_p,S'_p(\phi))$.

We wish to show that $(x_p,S'_p(\phi))$ is a vertex of
$\chi(F,\{x_1\})$.
Since $(x_p,S_p(\phi))$ is a vertex of
$\chi(F,\emptyset)$, 
$S_p(\phi)$ is a subset of $N_F(x_p)$.
So, by construction, $S'_p(\phi)$ is a subset of $N_F(x_p)$.
Each set $S_p(\phi)$ has even cardinality. 
Note that $N_P(x_p)$ has even cardinality unless $p=1$,
in which case it has odd cardinality.
Thus, $S'_p(\phi)$ has even cardinality unless $p=1$,
in which case it has odd cardinality.

The map~$b$ is a bijection since it can be inverted using $S_p(\phi) = S_p'(\phi) \oplus N_P(p)$.
The map $b(\phi)$ satisfies (E1) since $\phi$ satisfies (E1) 
and $x_{a} \in N_P(x_b)$ iff $x_{b} \in N_P(x_a)$.
The map $b(\phi)$  
It satisfies (E2) since the definition of $\calE(X,F,\emptyset,i)$ guarantees $j_i>1$
so $S_p'(\phi)\cap V_i^{j_i} = S_p(\phi) \cap V_i^{j_i}$.
It satisfies $c(\phi(x_p))=x_p$ for all $x_p\in X$ by construction.
Thus, $b(\phi) \in \calE(X,F,\{x_1\})$.
Finally, 
the same $j_i>1$ that shows $\phi \in \calE(X,F,\emptyset,i)$ shows that
$b(\phi)$ is in $\calE(X,F,\{x_1\},i)$.

\noindent{\bf Claim 2.} $|\calE(X,F,\emptyset,0)| > 0$.\quad

To prove Claim~2, we exhibit a $\phi \in \calE(X,F,\emptyset,0)$.  
Let $Z\subseteq Y$ be a set containing exactly one vertex from each component $C_1,\ldots,C_m$ 
such that every vertex in~$Z$ is adjacent to~$X$ in~$H$.
For each $z\in Z$, let  $p(z) = \min \{p\in [k] \mid (x_p,z)\in E(H)\}$.
For each $p\in [k]$
let $S_p =  \{z\in Z \mid p(z)= p\} \times \{2,\ldots,\ell\}$.
Since $\ell>1$ is odd, $|S_p|$ is even.
Let $\phi$ be the map from~$X$ to $V(\chi(F,\emptyset))$
such that, for all $p\in [k]$, $\phi(x_p) = (x_p,S_p)$.
The map~$\phi$ satisfies (E1) since $S_p\cap X = \emptyset$ for all $p\in [k]$.
It satisfies (E2) for any $i\in [m]$
by taking $j=1$ since for all $p\in [k]$, $|S_p \cap V_i^1|=0$.
It satisfies $c(\phi(x_p))=x_p$ for all $p\in [k]$ by construction so $\phi \in \calE(X,F,\emptyset)$.
To see that $\phi \in \calE(X,F,\emptyset,0)$, fix any $i\in[m]$ and any $j\in \{2,\ldots,\ell\}$.
Let $z$ be the unique vertex in $Z\cap V(C_i)$.
Then $(z,j)$ is the unique vertex in $S_{p(z)} \cap V_i^j$
Furthermore, for $p'\neq p(z)$ with $p'\in [k]$,
$S_{p'} \cap V_i^j$ is empty.
Thus, $\sum_{p=1}^k |S_p \cap V_i^j|$ is odd.

\noindent{\bf Claim 3.} $|\calE(X,F,\{x_1\},0)| = 0$.
 
We conclude by proving Claim~3. 
Assume for contradiction that   $\phi \in \calE(X,F,\{x_1\},0)$.
Let $G = \chi(F,\{x_1\})$.
As in Definition~\ref{def:extendable}, define $S_1,\ldots,S_k$ such that, for all $p\in [k]$, $\phi(x_p)=(x_p,S_p)$. Since $\phi(x_p)$ is in $V(G)$,
$|S_1|$ is odd and for every $p\in \{2,\ldots,k\}$, $|S_p|$ is even.
We will analyse a matrix $M$. 
The rows of $M$ are indexed by $x_1,\dots,x_k$ and the columns of $M$ are indexed by the vertices of $F$. For every vertex $v$ of~$F$, $M(x_p,v)$ is defined to be~$1$ if~$v\in S_p$ and $0$~otherwise.
Since $|S_1|$ is odd and $|S_2|,\dots,|S_k|$ are even, the total number of $1$s in $M$ is odd.

We split $M$ into two submatrices: $M_X$ consists of the columns indexed by vertices in~$X$ and $\MD$   consists of the columns indexed by the remaining vertices of~$F$.
For every $p\in [k]$, $x_p\notin S_p$ since $(x_p,S_p)$ is a vertex of~$G$ and~$F$ has no self-loops. Also, $M_X$ is symmetric by (E1). So the number of $1$s in $M_X$ is even.

Now consider $\MD$. Since $\phi \in \calE(X,F,\{x_1\})$, 
(E2) implies that for every $i\in [m]$ there is a $j\in [\ell]$ such that
$\sum_{p=1}^k |S_p \cap V_i^j|$ is even.
Since $\phi\in \calE(X,F,\{x_1\},0)$,
for each $i\in[m]$, the following conditions hold.
\begin{itemize}
\item[(C1)] $\sum_{p=1}^k |S_p \cap V_i^1| $ is even, and
\item[(C2)] for all $j>1$,  $\sum_{p=1}^k |S_p \cap V_i^j| $ is
odd. 
    \end{itemize}

For each $i\in[m]$ and $j\in[\ell]$, let $M^j_i$ be the submatrix of $\MD$ containing only the columns indexed by the vertices in $V_i^j$. Then 
for all $i\in [m]$, condition~(C1) implies that $M_i^1$ contains an even number of $1$s   and 
condition~(C2) implies that for all $j\in \{2,\ldots,\ell\}$,
$M_i^j$ contains an odd number of $1$s.  Since $\ell$ is odd,
the total number of $1$s in the matrices 
$M_i^1,\ldots,M_i^\ell$ is even.
Consequently, the total number of $1$s in $\MD$ is even, contradicting the fact that $M$ has an odd number of $1$s and $M_X$ has an even number of $1$s.  \end{proof}

\section{Main Result and Consequences}\label{sec:main}
With upper and lower bounds established, we are now able to proof Theorem~\ref{thm:main_thm}, which we restate for convenience.

\mthm*
\begin{proof}
We first consider the special case where $(H,X)$ is a \emph{full conjunctive queries}, that is, no variable of $(H,X)$ is existentially quantified  so $X=V(H)$. In this case, $(H,X)$ is counting minimal, since counting equivalence is the same as isomorphism in this case~\cite{DellRW19}. Moreover, $\Gamma(H,X)=H$. Thus $\sew(H,X)=\tw(H)$. Since $\Ans((H,X),G)=\Hom(H,G)$ for $X=V(H)$, counting answers to $(H,X)$ is the same as counting homomorphisms from $H$, and the WL-dimension of counting homomorphisms is $\tw(H)$ as shown by Neuen~\cite{neuen2023homomorphism}.

Now consider the case where $X \neq V(H)$ and let $(H',X')$ be a counting minimal conjunctive query with $(H',X')\sim (H,X)$. Then, $|\Ans((H,X),G)|=|\Ans((H',X'),G)|$ for every graph $G$ and thus $(H,X)$ and $(H',X')$ have the same WL-dimension. 
Furthermore, since $(H,X)$ is connected, so is $(H',X')$ --- see~\cite[Section 6]{DellRW19arxiv}. 
Theorems~\ref{thm:lower_bound} and~\ref{thm:upper_bound} now state that the WL-dimension of $(H',X')$ is equal to $\ctw(H',X')$. Finally, by definition of semantic extension width, we have $\sew(H,X)=\ctw(H',X')$, concluding the proof.
\end{proof}

\subsection{Homomorphism Indistinguishability and Conjunctive Queries}
Given a class of graphs $\mathcal{F}$, two graphs $G$ and $G'$ are called $\mathcal{F}$\emph{-indistinguishable}, denoted by $G \cong_{\mathcal{F}} G'$ if $|\Hom(F,G)|=|\Hom(F,G')|$ for all $F\in \mathcal{F}$. If $\mathcal{F}$ is the class of all graphs, then a classical result of Lov{\'{a}}sz states that $\cong_\mathcal{F}$ coincides with isomorphism (see e.g.\ Theorem 5.29 in~\cite{Lovasz12}). Recent years have seen numerous exciting results on the structure of $\mathcal{F}$-indistinguishability, depending on the class $\mathcal{F}$: For example, Dvor{\'{a}}k~\cite{Dvorak10}, and Dell, Grohe and Rattan~\cite{DellGR18} have shown that $\cong_\mathcal{F}$ coincides with $\cong_k$, i.e., with $k$-WL-equivalence, if $\mathcal{F}$ is the class of all graphs of treewidth at most $k$, and Mancinska and Roberson have shown that $\cong_\mathcal{F}$ coincides with what is called quantum-isomorphism if $\mathcal{F}$ is the class of all planar graphs~\cite{MancinskaR20}.

To state our first corollary, we extend the notion of homomorphism indistinguishability to conjunctive queries.

\begin{definition}
    Let $\Psi$ be a class of conjunctive queries. Two graphs $G$ and $G'$ are $\Psi$-\emph{indistinguishable}, denoted by $G \cong_\Psi G'$, if $|\Ans((H,X),G)|=|\Ans((H,X),G')|$ for all queries $(H,X)\in \Psi$.
\end{definition}

Then, using the notion of conjunctive query indistinguishability, we obtain a new characterisation of $k$-WL-equivalence.

\begin{corollary}[Corollary~\ref{cor:into_WL_char}, restated]
Let $k$ be a positive integer and let $\Psi_k$ be the set of all connected conjunctive queries with at least one free variable and with semantic extension width at most $k$. Then for any pair of graphs $G$ and $G'$,   $G\cong_k G'$ if and only if $G \cong_{\Psi_k} G'$.
\end{corollary}
\begin{proof}
For the first direction, suppose that $G\cong_k G'$ and consider $(H,X)\in \Psi_k$. Then $\sew(H,X)\leq k$ and thus the WL-dimension of $(H,X)$ is at most $k$ by Theorem~\ref{thm:main_thm}. Consequently, $|\Ans((H,X),G)|=|\Ans((H,X),G')|$. This shows that $G \cong_{\Psi_k} G'$. 

For the other direction, suppose that $G \cong_{\Psi_k} G'$. Recall that $\sew(H,V(H))=\tw(H)$ for all $h$. Thus $\Psi_k$ contains all conjunctive queries $(H,V(H))$ with $\tw(H)\leq k$. Since $\Ans((H,V(H)),F)=\Hom(H,F)$ for all $H$ and $F$, $G \cong_{\Psi_k} G'$ implies $G \cong_{\mathcal{F}'} G'$ where $\mathcal{F}'$  is the 
class of all conjunctive queries $(H,V(H))$ where $H$ is a 
\emph{connected} graph with   treewidth at most $k$. Finally, we can remove the connectivity constraint as follows: Let $F=F_1 \cup F_2$ be the disjoint union of two graphs $F_1$ and $F_2$. If $F$ has treewidth at most $k$, then both $F_1$ and $F_2$ also have treewidth at most $k$, since treewidth is monotone under taking subgraphs. Thus, $G \cong_{\mathcal{F}'} G'$ implies  $\Hom(F,G)=\Hom(F_1,G)\cdot \Hom(F_2,G) = \Hom(F_1,G')\cdot \Hom(F_2,G')= \Hom(F,G')$. Consequently $G \cong_{\mathcal{F}} G'$ where $\mathcal{F}$ 
is the class of all queries $(H,V(H))$ such that $H$ is a graph with treewidth at most $k$  and thus $G\cong_k G'$.
\end{proof}

In the following corollary, we show that the treewidth of a conjunctive query alone is insufficient for describing the WL-dimension. This is even the case for treewidth $1$, i.e., for acyclic queries.
\begin{corollary}\label{thm:acyclic_WL_dim}
    The class of acyclic conjunctive queries have unbounded WL-dimension, that is, there is no $k$ such that $G \cong_k G'$ if and only if $G \cong_{\mathcal{T}} G'$, where $\mathcal{T}$ is the class of all acyclic conjunctive queries.
\end{corollary}
\begin{proof}
The corollary follows immediately from Theorem~\ref{thm:main_thm} and the fact that acyclic conjunctive queries can have arbitrary high semantic extension width. Recall the $k$-star query $(S_k,X_k)$, where $X_k=\{x_1,\dots,x_k\}$, $V(S_k)=X_k \cup \{y\}$ and $E(S_k)=\{\{x_i,y\}\mid i\in[k] \}$. Clearly, $(S_k,X_k)$ is acyclic. Moreover, it is well-known that $(s_k,X_k)$ is counting minimal (see e.g. \cite{DellRW19}). Finally, $\Gamma(S_k,X_k)$ is the $k+1$-clique, and thus $\sew(S_k,X_k)=\tw(K_{k+1})=k$.
\end{proof}

Corollary~\ref{thm:acyclic_WL_dim}
is in stark contrast to the quantifier-free case. The WL-dimension of any acyclic conjunctive query $(H,V(H))$  is equal to $1$ since this case is equivalent to counting homomorphisms from acyclic graphs~\cite{Dvorak10,DellGR18}. 
Given Corollary~\ref{thm:acyclic_WL_dim} one might ask how powerful indistinguishability by acyclic conjunctive queries is: Is there any $k>1$ such that $\mathcal{T}$-indistinguishability is at least as powerful as $k$-WL-equivalence? We provide a negative answer.

\begin{observation}
Let $2K_3$ denote the graph consisting of two disjoint triangles and let $C_6$ denote the $6$-cycle. Let $(H,X)$ be a connected and acyclic conjunctive query. Then $|\Ans((H,X),2K_3)| = |\Ans((H,X),C_6|$.
\end{observation}
\begin{proof}
For the disjoint union of two conjunctive queries $(H_1,X_1)\cup (H_2,X_2)$ and any graph $G$, $|\Ans((H_1,X_1)\cup (H_2,X_2),G)|=|\Ans((H_1,X_1),G)|\cdot|\Ans((H_2,X_2),G)|$. Thus, it suffices to prove the observation for connected queries.
    
Let $(H,X)$ be any connected acyclic conjunctive query. If $X=\emptyset$ then the observation is trivial since computing $G \mapsto |\Ans((H,X),G)|$ is equivalent to deciding whether there is a homomorphism from $H$ to $G$. As $H$ is acyclic, we thus have $|\Ans((H,X),2K_3)| = |\Ans((H,X),C_6| = 1$.
    
Now assume that $X\neq \emptyset$. For the proof, we associate the query $(H,X)$ with an edge-weighted tree $T$ as follows: The vertices of $T$ are $X$, and two (distinct) vertices $x_1,x_2$ of $T$ are adjacent if and only if there is a path from $x_1$ to $x_2$ in $H$, the intermediate vertices of which are all existentially quantified variables, i.e., contained in $V(H)\setminus X$. Note that there can be at most one such path since $(H,X)$ is acyclic. The number of intermediate vertices on this path will become the weight of the edge $\{x_1,x_2\}$ in $T$; note that the weight is $0$ if and only if $x_1$ and $x_2$ are adjacent in $H$. Let us write $w:E(T)\to \mathbb{N}$ for the weight function of the edges of $T$.

Now it is easy to see that, for any graph $G$ without isolated vertices, the elements of $\Ans((H,X),G)$ are precisely the mappings $\varphi: V(T) \to V(G)$ such that for every $e=\{x_1,x_2\}\in E(T)$ there is a (not necessarily simple) walk from $\varphi(x_1)$ to $\varphi(x_2)$ with $w(e)$ internal vertices.

Let us write $\#\Ans((T,w),G)$ for the set of such mappings. It remains to show that $\#\Ans((T,w),C_6)=\#\Ans((T,w),2K_3)$. We prove this claim by induction on $n=|V(T)|$. Since $V(T)=X$ and $X\neq \emptyset$, the induction base is $n=1$, for which we have that $\#\Ans((T,w),C_6)=\#\Ans((T,w),2K_3) = 6$.

For the induction step, fix a leaf $x$ of $T$, let $x'$ be its neighbour, and let $\hat{w}=w(\{x,x'\})$. Moreover, let $T'=T\setminus \{x\}$ and $w'=w|_{E(T')}$. By the induction hypothesis we have $\#\Ans((T',w'),C_6)=\#\Ans((T',w'),2K_3)$. Now fix \emph{any} pair of queries $\varphi_1\in \Ans((T',w'),C_6)$ and $\varphi_2 \in \Ans((T',w'),2K_3)$. Let $s_1$ and $s_2$ be the number of extensions of $\varphi_1$ and $\varphi_2$ that yield an element in $\Ans((T,w),C_6)$ and $\Ans((T,w),2K3)$, respectively, that is, $s_1=|V_1|$ and $s_2=|V_2|$ where
\begin{align*}
    V_1 &=\{v \in V(C_6)\mid \varphi_1 \cup \{x \mapsto v\}\in \Ans((T,w),C_6) \}\\
    V_2 &=\{v \in V(2K_3)\mid \varphi_1 \cup \{x \mapsto v\}\in \Ans((T,w),2K_3) \}
\end{align*}
Next we claim
\begin{align*}
    s_1=s_2=\begin{cases}
        2 & \hat{w}=0\\
        3 & \hat{w}>0
    \end{cases}
\end{align*}
Note that proving this claim concludes the proof since it implies that \[\#\Ans((T,w),C_6)= s\cdot \#\Ans((T',w'),C_6)= s\cdot \#\Ans((T',w'),2K_3)= \#\Ans((T,w),2K3),\]
where $s\in \{2,3\}$ depends only on $\hat{w}$.

Finally, to prove the claim, assume $C_6$ has vertices $\{0,\dots,5\}$ and edges $\{i,i+1 \mod 6\}$. Moreover, assume that $2K3$ has vertices $\{a,b,c,a',b',c'\}$ and that the triangles are $\{a,b,c\}$ and $\{a',b',c'\}$.
Now let $u_1=\varphi_1(x')$ and $u_2=\varphi_2(x')$. W.l.o.g.\ assume that $u_1=0$ and $u_2=a$. 
\begin{itemize}
    \item If $\hat{w}=0$ then $x$ must be mapped to a neighbour of the image of $x'$. Since both $2K3$ and $C_6$ are $2$-regular, there are precisely $2$ options for each case.
    \item If $\hat{w}$ is odd, then $V_1=\{0,2,4\}$ and $V_2=\{a,b,c\}$. Thus $n_1=n_2=3$.
    \item Otherwise $\hat{w}$ is positive and even. Then $V_1=\{1,3,5\}$ and $V_2=\{a,b,c\}$. Thus $n_1=n_2=3$.
\end{itemize}
This concludes the proof.
\end{proof}
In other words, acyclic conjunctive queries cannot even distinguish $2K_3$ and $C_6$, which are the most common examples of graphs which are $1$-WL-equivalent, but which are not $2$-WL-equivalent.

\subsection{WL-Dimension and the Complexity of Counting}
 
In this section, we give a connection between   WL-dimension and the parameterised complexity of counting answers to conjunctive queries. Recall that, given a class of conjunctive queries $\Psi$, the counting problem $\#\textsc{CQ}(\Psi)$ takes as input a pair consisting of a conjunctive query $(H,X)\in \Psi$ and a graph $G$ and outputs $|\Ans((H,X),G)|$. 

Recall that a class of conjunctive queries has \emph{bounded} WL-dimension if there is a constant $B$ that upper bounds the WL-dimension of all queries in the class. 
\corcomplexity*
\begin{proof}

Given a conjunctive query $(H,X)$, recall the definition of the graph $\Gamma(H,X)$ from Definition~\ref{def:ew}.
The \emph{contract} of   $(H,X)$ is the 
induced graph subgraph~$\Gamma[X]$ (see~\cite[Definition 8]{DellRW19}).

Using the notion of contract,
Chen, Durand, and Mengel~\cite{DurandM15,ChenM15} and Dell, Roth and Wellnitz~\cite{DellRW19} have established an exhaustive complexity classification for $\#\textsc{CQ}(\Psi)$. For an explicit statement, Theorem 10 in~\cite{DellRW19} states that (conditioned on $\mathrm{FPT} \neq W[1]$),
$\#\textsc{CQ}(\Psi)$ is solvable in polynomial time if and only if 
both the treewidth of $\Psi$ and the treewidth of the \emph{contracts} of $\Psi$ are  bounded.

By Theorem~\ref{thm:main_thm}, the WL-dimension of a 
connected counting minimal query with $X\neq \emptyset$ 
is equal to its extension width.
It remains to show that the condition that both the treewidth of $\Psi$ and the treewidth of the contracts of $\Psi$ are bounded
is the same as the condition that $\Psi$ has bounded extension width.

First, suppose that $\Psi$ has bounded extension width.
Then for some quantity~$B$,
every query $(H,X)$ in $\Psi$ has extension width at most~$B$.
So $\tw(\Gamma(X,H))\leq B$. Since $H$ and $\Gamma[X]$ are both subgraphs of
$\Gamma(X,H)$, they both have treewidth at most~$B$, as required.

For the other direction, suppose that the treewidth of every query in~$\Psi$ and every contract of every query in~$\Psi$ is at most~$B$.
Consider a query $(H,X) \in \Psi$. Let $t_1\leq B$ be the treewidth of $H$ and let $t_2\leq B$ be the treewidth of $\Gamma[X]$. Let $\mathcal{T}=(T,\mathcal{B})$ be an optimal tree-decomposition of $\Gamma[X]$. Let $C_1,\dots,C_d$ be the connected components of $H[Y]$. For each $i\in[d]$ let $\delta_i$ be the subset of vertices in $X$ that are adjacent to $C_i$ in $H$. By the definition of $\Gamma$, each $\delta_i$ induces a clique in $\Gamma$, and thus in $\Gamma[X]$. Therefore, there is a bag $B_i$ that  contains $\delta_i$ (see e.g.\ \cite[Lemma 3]{BodlaenderK06}). Since the treewidth of $\Gamma[X]$ is $t_2$ it follows that $|\delta_i|\leq t_2+1$. 

Next, for each $i\in[d]$, let $\mathcal{T}_i$ be an optimal tree decomposition of $C_i$ and note that the width of all $\mathcal{T}_i$ is at most $t_1$ since all components~$C_i$ are subgraphs of $H$. It is now straightforward to construct a tree-decomposition of width at most $t_1+t_2$ of $\Gamma(H,X)$. For every $i\in[d]$, we add $\delta_i$ to each bag of $\mathcal{T}_i$. Finally, we fix an arbitrary node $v_i$ of the tree $T_i$ of $\mathcal{T}$ and connect it to the node of $T$ with bag $B_i$. Clearly (T1)-(T3) are satisfied. 
 
\end{proof}

\subsection{Linear Combinations of Conjunctive Queries}

The study of linear combinations of homomorphism counts dates back to the work of Lov{\'{a}}sz (see  the textbook~\cite{Lovasz12}). Recently, staring with the work of Chen and Mengel~\cite{ChenM16} and of Curticapean, Dell and Marx~\cite{CurticapeanDM17}, the study of these linear combinations has re-arisen in the context of parameterised counting complexity theory. Moreover, the works of Seppelt~\cite{Seppelt23}, Neuen~\cite{neuen2023homomorphism}, and Lanzinger and Barceló~\cite{LanzingerB23} have shown that the WL-dimension of a function evaluating a finite linear combination of homomorphism counts is equal to the maximum WL-dimension of any term in the combination. Using Theorem~\ref{thm:main_thm}, we  establish a similar  result (Corollary~\ref{cor:quantum_WL}) for linear combinations of conjunctive queries. 
This gives a precise quantification   of
the WL-dimension of unions of conjunctive queries and of conjunctive queries with disequalities and negations over the free variables.

Following Lov{\'{a}}sz's notion of a ``quantum graph''\cite[Chapter 6.1]{Lovasz12}, we formalise our linear combinations as follows.
\begin{definition}[Quantum Query]\label{def:quantumquery}
A \emph{quantum query} $Q$ is a formal finite linear combination of conjunctive queries
$Q = \sum_{i=1}^\ell c_i\cdot (H_i,X_i)$
such that, for all $i\in[\ell]$, $c_i \in \mathbb{Q}\setminus\{0\}$ and $(H_i,X_i)$ is a connected and counting minimal conjunctive query with $X_i\neq 0$. Moreover, the 
conjunctive queries
$(H_i,X_i)$ are pairwise non-isomorphic. We call the queries $(H_i,X_i)$ the \emph{constituents} of $Q$.
The number of answers of $Q$ in a graph $G$ is defined as  
$|\Ans(Q,G)| := \sum_{i=1}^\ell c_i\cdot |\Ans((H_i,X_i),G)|$.
\end{definition}

Chen and Mengel~\cite{ChenM16}, and Dell, Roth, and Wellnitz~\cite{DellRW19} have shown that 
for every 
union $\phi$  of (connected) conjunctive queries with at least one free variable
there is a quantum query $Q[\phi]$ such that,
for all graphs $G$, the number of answers of $\varphi$ in $G$ is equal to $|\Ans(Q[\varphi],G)|$. Moreover, $Q[\varphi]$ is unique up to reordering terms (and up to isomorphim of the constituents). 
They have also shown a similar result
when $\phi$ is a conjunctive query with disequalities and negations.

\begin{definition}[Hereditary Semantic Extension Width]
    The \emph{hereditary semantic extension width} of a quantum query $Q = \sum_{i=1}^\ell c_i\cdot (H_i,X_i)$
    is $\hsew(Q)=\max\{\sew(H_i,X_i) \mid i \in[\ell] \}$.
\end{definition}

We define the WL-dimension of a quantum query $Q$ as the WL-dimension of the graph parameter $G \mapsto |\Ans(Q,G)|$. The following lemma was shown by Seppelt~\cite{Seppelt23} in the special case of homomorphisms, i.e., the case in which each constituent $(H_i,X_i)$ satisfies $X_i =V(H_i)$. The proof of the generalised version follows the same idea, combining   Seppelt's approach  and Theorem~\ref{thm:main_thm}.

\corquantum*
\begin{proof}
Let $Q = \sum_{i=1}^\ell c_i\cdot (H_i,X_i)$ and $k=\hsew(Q)$. We first show that the WL-dimension of $Q$ is upper bounded by $k$; this is the easy direction. 
For each $i\in [\ell]$,
let $k_i = \sew(H_i,X_i)$. 
By the definition of $\hsew(Q)$, 
$k= \max_{i\in [\ell]} k_i$.

By Theorem~\ref{thm:main_thm}, the WL-dimension of the function $G\mapsto |\Ans(H_i,X_i)|$ is~$k_i$.  
Let $G$ and $G'$ be graphs such that
$G\cong_k G'$ (which means that for every $i\in [\ell]$, $G\cong_{k_i} G'$).
Then for every $i\in [\ell]$,  
$|\Ans((H_i,X_i),G)|=|\Ans((H_i,X_i),G')|$
so
\[|\Ans(Q,G)| = \sum_{i=1}^\ell c_i\cdot |\Ans((H_i,X_i),G)| = \sum_{i=1}^\ell c_i\cdot |\Ans((H_i,X_i),G')| =|\Ans(Q,G')|. \]
This means that the function 
$G\mapsto |\Ans(Q,G)|$ cannot distinguish $k$-WL-equivalent graphs so the WL-dimension of~$Q$ is at most~$k$.

The more difficult direction is the lower bound. To this end, we will construct graphs $F$ and $F'$ such that $F\cong_{k-1} F'$ and $|\Ans(Q,F)|\neq |\Ans(Q,F')|$. 
This implies that the WL-dimension of 
the graph parameter $G\mapsto |\Ans(Q,G)|$ is greater than $k-1$ so the WL-dimension of $Q$ is at least~$k$.

Assume without loss of generality that $\sew(H_1,X_1)= k$. By Theorem~\ref{thm:main_thm} the WL-dimension of $(H_1,X_1)$ is $k$ 
so $k$ is the minimum positive integer such that counting answers from $(H_1,X_1)$ cannot
distinguish $k$-equivalent graphs.
Thus there are graphs $G$ and $G'$ with $G\cong_{k-1} G'$ and $|\Ans((H_1,X_1),G)|\neq |\Ans((H_1,X_1),G')|$. 

The   \emph{tensor product}  of two graphs $A$ and $B$,   denoted by $A \otimes B$, has the property that, for every graph $H$,  $|\Hom(H,A\otimes B)|=|\Hom(H,A)|\cdot|\Hom(H,B)|$. See, for example, \cite[Chapters 3.3 and 5.2.3]{Lovasz12}. 
For each graph $T$ of treewidth at most $k-1$,   $|\Hom(T,G\otimes H)|=|\Hom(T,G)|\cdot|\Hom(T,H)| = |\Hom(T,G')|\cdot|\Hom(T,H)| = |\Hom(T,G'\otimes H)|$ Thus,
for all graphs $H$,  $G \otimes H \cong_{k-1} G'\otimes H$.

Suppose for contradiction that, 
for all $H$,
$|\Ans(Q,G \otimes H)|= |\Ans(Q,G'\otimes H)|$. This rewrites to
\[ \sum_{i=1}^\ell c_i(|\Ans((H_i,X_i),G)|-|\Ans((H_i,X_i),G')|) \cdot |\Ans((H_i,X_i),H)| = 0. \]
Let $d_i:= c_i(|\Ans((H_i,X_i),G)|-|\Ans((H_i,X_i),G')|)$ and  $m=\max\{|V(H_i)|\mid i\in[\ell] \}$. Let  $\mathcal{H}$ be the set of all graphs with at most $m^m$ vertices.
Then we obtain a system of linear equations
containing the following equation for each $H\in \mathcal{H}$.
\[\sum_{i=1}^\ell d_i\cdot  |\Ans((H_i,X_i),H)| = 0\,. \]
It was shown in~\cite[Lemma~34(iii)]{DellRW19arxiv} that the matrix corresponding to this system has full rank. Therefore, for all $i\in [\ell]$,
$d_i=0$. This implies $|\Ans((H_1,X_1),G)|=|\Ans((H_1,X_1),G')|$, contradicting the choice of~$G$ and~$G'$. Thus our assumption was wrong, and there is a graph $H$ (in fact $H\in\mathcal{H}$) such that $|\Ans(Q,G \otimes H)|\neq |\Ans(Q,G'\otimes H)|$. Since $G \otimes H \cong_{k_1} G' \otimes H$, the proof is concluded.
\end{proof}

\subsection{Star Queries and Dominating Sets}
In this final section we  use Theorem~\ref{thm:main_thm} to determine the WL-dimension of counting dominating sets (proving Corollary~\ref{cor:intro_domset}).

\begin{definition}[Dominating Set]
Let $G$ be a graph and let $k$ be a positive integer. A \emph{dominating set} of $G$ is a subset $D\subseteq V(G)$ such that each vertex of $G$ is either contained in~$D$ or adjacent to a vertex in $D$. The set $\Delta_k(G)$ contains all size-$k$ dominating sets of  $G$.
\end{definition}

For analysing the WL-dimension of counting size-$k$ dominating sets   we will consider, as an intermediate step, the $k$-star-query.
\begin{definition}
Let $k$ be a positive integer. The $k$-\emph{star} is the conjunctive query $(S_k,X_k)$ where $X_k=\{x_1,\dots,x_k\}$, $V(S_k)=X \cup \{y\}$, and $E(S_k)=\{\{x_i,y\} \mid i\in[k] \}$.
\end{definition}

The $k$-star is often written in the more prominent form
$\phi(x_1,\dots,x_k) =  \exists y: \bigwedge_{i=1}^k E(x_i,y).$
It is well-known that $(S_k,X_k)$ is counting minimal (see e.g.~\cite{DellRW19}). Moreover, $\Gamma(S_k,X_k)$ is the $(k+1)$-clique, which has treewidth $k$. Thus $\sew(S_k,X_k)= k$.
Corollary~\ref{cor:cor_star} follows immediately 
from   Theorem~\ref{thm:main_thm}.
\begin{corollary}\label{cor:cor_star}
    The WL-dimension of $(S_k,X_k)$ is $k$.
\end{corollary}

We can now prove Corollary~\ref{cor:intro_domset}. 

\begin{corollary}[Corollary~\ref{cor:intro_domset}, restated]\label{cor:cor_ds}
    The WL-dimension of the function $G \mapsto |\Delta_k(G)|$ is $k$.
\end{corollary}
\begin{proof}
\newcommand{\inj}{\mathsf{Inj}}
We start with the lower bound. To this end, given a graph $G$, and a conjunctive query $(H,X)$, we set
\[\inj((H,X),G)=\{a \in \Ans((H,X),G) \mid a \text{ is injective} \}.\]
Let $I =\{(i,j) \in[k]^2\mid i<j \}$ and consider a subset $J\subseteq I$. The query $(S_k,X_k)/J$ is obtained by identifying $x_i$ and $x_j$ if and only if $(i,j)\in J$. Observe that $(S_k,X_k)/J \cong (S_\ell,X_\ell)$ for some $\ell \leq k$. By the principle of inclusion and exclusion, for each graph $G$,
\[|\inj((S_k,X_k),G)| = \sum_{J \subseteq I} (-1)^{|J|}\cdot  |\Ans((S_k,X_k)/J,G)| = \sum_{i=1}^k c_i \cdot |\Ans((S_i,X_i),G)|\,,\]
where $c_i = \{J \subseteq I \mid (S_k,X_k)/J \cong (S_i,X_i) \}$. Thus, $|\inj((S_k,X_k),G)|$ computes the number of answers to the quantum query with constituents $(S_i,X_i)$ and coefficients $c_i$. Moreover, $c_k=1$ since $(S_k,X_k)/J \cong (S_k,X_k)$ if and only if $J = \emptyset$.\footnote{To obtain a quantum query, we need to remove all terms $(S_i,X_i)$ with $c_i=0$. In fact, following the analysis in~\cite{CurticapeanDM17}, it can be shown that none of the $c_i$ is $0$. However, since we only need $c_k\neq 0$, we omit going into further details.} By Corollary~\ref{cor:quantum_WL} and the fact that $\sew(S_\ell,X_\ell)=\ell$ for all $\ell\in[k]$, the WL-dimension of  $G \mapsto |\inj((S_k,X_k),G)|$ is equal to $k$.

Next observe that $|\inj((S_k,X_k),G)|/k!$ is the number of $k$-vertex subsets $A$ of $G$ such that there is a vertex $y\in V(G)$ that is adjacent to all $a\in A$. Thus $\binom{|V(G)|}{k} - |\inj((S_k,X_k),G)|/k!$ is equal to the size of the set
\[  D_k(G) := \{ A \subseteq V(G) \mid |A|=k ~\wedge~ \forall y \in V(G) : \exists a\in A: \{a,y\} \notin E(G) \} . \]
Let $\overline{G}$ be the self-loop-free complement of $G$, that is, two \emph{distinct} vertices $u$ and $v$ in $V(\overline{G})=V(G)$ are adjacent in $\overline{G}$ if and only if they are not adjacent in $G$. Observe that $\{a,y\} \notin E(G)$ if and only if $a = y$ or $\{a,y\} \in E(\overline{G})$. Therefore $|D_k(G)|=|\Delta_k(\overline{G})|$.

We are now ready to prove that the WL-dimension of the function $G \mapsto |\Delta_k(G)|$ is at least $k$. Suppose for contradiction that its WL-dimension is $k'$ for some $1\leq k' <k$. Then,  for all $G$ and $G'$ with $G \cong_{k'} G'$,   $|\Delta_k(G)|=|\Delta_k(G')|$. However, we know that the WL-dimension of $G \mapsto |\inj((S_k,X_k),G)|$ is equal to $k$. Thus there are 
graphs $F$ and $F'$ with
$F \cong_{k'} F'$ and $|\inj((S_k,X_k),F)| \neq |\inj((S_k,X_k),F')|$.   
It is well known (see e.g.\ Seppelt~\cite{Seppelt23}) that $F \cong_{k'} F'$ implies $\overline{F} \cong_{k'} \overline{F'}$. Therefore $|\Delta_k(\overline{F})|=|\Delta_k(\overline{F'})|$. 
Let $K_1$ be the (treewidth~$0$) graph containing one isolated vertex. 
The number of homomomorphisms from~$K_1$ to~$F$
determines the number of vertices of $F$
so  $F \cong_{k'} F'$ implies $|V(F)|=|V(F')|$.
Let $n=|V(F)|=|V(F')|$. In summary, 
\[|\inj((S_k,X_k),F)| = k!\left(\binom{n}{k} - |\Delta_k(\overline{F})|\right) = k!\left(\binom{n}{k} - |\Delta_k(\overline{F'})|\right) = |\inj((S_k,X_k),F')|, \]
which contradicts the choice of $F$ and $F'$ and concludes the proof of the lower bound.

For the upper bound, we have to show that $F \cong_k F'$ implies $\Delta_k(F)=\Delta_k(F')$, which is an immediate consequence of our previous analysis. Since the WL-dimension of $G \mapsto |\inj((S_k,X_k),G)|$ is equal to $k$, we have 
\[|\Delta_k(F)| = \binom{n}{k} - |\inj((S_k,X_k),\overline{F})|/k! = \binom{n}{k} - |\inj((S_k,X_k),\overline{F'})|/k! = |\Delta_k(F')|.\]
This concludes the proof.
\end{proof}

\bibliographystyle{plainurl}
\bibliography{references.bib}

\begin{thebibliography}{10}

\bibitem{Alonetal08}
Noga Alon, Phuong Dao, Iman Hajirasouliha, Fereydoun Hormozdiari, and S.~Cenk
  Sahinalp.
\newblock {Biomolecular network motif counting and discovery by color coding}.
\newblock {\em Bioinformatics}, 24(13):i241--i249, 07 2008.

\bibitem{Arvind16}
Vikraman Arvind.
\newblock {The Weisfeiler-Lehman Procedure}.
\newblock {\em Bull. {EATCS}}, 120, 2016.
\newblock URL: \url{http://eatcs.org/beatcs/index.php/beatcs/article/view/442}.

\bibitem{ArvindFKV22}
Vikraman Arvind, Frank Fuhlbr{\"{u}}ck, Johannes K{\"{o}}bler, and Oleg
  Verbitsky.
\newblock On the weisfeiler-leman dimension of fractional packing.
\newblock {\em Inf. Comput.}, 288:104803, 2022.
\newblock \href {https://doi.org/10.1016/j.ic.2021.104803}
  {\path{doi:10.1016/j.ic.2021.104803}}.

\bibitem{BarceloGRR21}
Pablo Barcel{\'{o}}, Floris Geerts, Juan~L. Reutter, and Maksimilian Ryschkov.
\newblock Graph neural networks with local graph parameters.
\newblock In Marc'Aurelio Ranzato, Alina Beygelzimer, Yann~N. Dauphin, Percy
  Liang, and Jennifer~Wortman Vaughan, editors, {\em Advances in Neural
  Information Processing Systems 34: Annual Conference on Neural Information
  Processing Systems 2021, NeurIPS 2021, December 6-14, 2021, virtual}, pages
  25280--25293, 2021.

\bibitem{Bodlaender03}
Hans~L. Bodlaender.
\newblock Necessary edges in k-chordalisations of graphs.
\newblock {\em J. Comb. Optim.}, 7(3):283--290, 2003.
\newblock \href {https://doi.org/10.1023/A:1027320705349}
  {\path{doi:10.1023/A:1027320705349}}.

\bibitem{BodlaenderK06}
Hans~L. Bodlaender and Arie M. C.~A. Koster.
\newblock Safe separators for treewidth.
\newblock {\em Discret. Math.}, 306(3):337--350, 2006.
\newblock \href {https://doi.org/10.1016/j.disc.2005.12.017}
  {\path{doi:10.1016/j.disc.2005.12.017}}.

\bibitem{Boker19}
Jan B{\"{o}}ker.
\newblock Color refinement, homomorphisms, and hypergraphs.
\newblock In Ignasi Sau and Dimitrios~M. Thilikos, editors, {\em
  Graph-Theoretic Concepts in Computer Science - 45th International Workshop,
  {WG} 2019, Vall de N{\'{u}}ria, Spain, June 19-21, 2019, Revised Papers},
  volume 11789 of {\em Lecture Notes in Computer Science}, pages 338--350.
  Springer, 2019.
\newblock \href {https://doi.org/10.1007/978-3-030-30786-8\_26}
  {\path{doi:10.1007/978-3-030-30786-8\_26}}.

\bibitem{BouritsasFZB23}
Giorgos Bouritsas, Fabrizio Frasca, Stefanos Zafeiriou, and Michael~M.
  Bronstein.
\newblock Improving graph neural network expressivity via subgraph isomorphism
  counting.
\newblock {\em {IEEE} Trans. Pattern Anal. Mach. Intell.}, 45(1):657--668,
  2023.
\newblock \href {https://doi.org/10.1109/TPAMI.2022.3154319}
  {\path{doi:10.1109/TPAMI.2022.3154319}}.

\bibitem{ButtiD21}
Silvia Butti and V{\'{\i}}ctor Dalmau.
\newblock {Fractional Homomorphism, Weisfeiler-Leman Invariance, and the
  Sherali-Adams Hierarchy for the Constraint Satisfaction Problem}.
\newblock In Filippo Bonchi and Simon~J. Puglisi, editors, {\em 46th
  International Symposium on Mathematical Foundations of Computer Science,
  {MFCS} 2021, August 23-27, 2021, Tallinn, Estonia}, volume 202 of {\em
  LIPIcs}, pages 27:1--27:19. Schloss Dagstuhl - Leibniz-Zentrum f{\"{u}}r
  Informatik, 2021.
\newblock \href {https://doi.org/10.4230/LIPIcs.MFCS.2021.27}
  {\path{doi:10.4230/LIPIcs.MFCS.2021.27}}.

\bibitem{CaiFI92}
Jin{-}yi Cai, Martin F{\"{u}}rer, and Neil Immerman.
\newblock An optimal lower bound on the number of variables for graph
  identification.
\newblock {\em Comb.}, 12(4):389--410, 1992.
\newblock \href {https://doi.org/10.1007/BF01305232}
  {\path{doi:10.1007/BF01305232}}.

\bibitem{ChenM15}
Hubie Chen and Stefan Mengel.
\newblock A trichotomy in the complexity of counting answers to conjunctive
  queries.
\newblock In Marcelo Arenas and Mart{\'{\i}}n Ugarte, editors, {\em 18th
  International Conference on Database Theory, {ICDT} 2015, March 23-27, 2015,
  Brussels, Belgium}, volume~31 of {\em LIPIcs}, pages 110--126. Schloss
  Dagstuhl - Leibniz-Zentrum f{\"{u}}r Informatik, 2015.
\newblock \href {https://doi.org/10.4230/LIPIcs.ICDT.2015.110}
  {\path{doi:10.4230/LIPIcs.ICDT.2015.110}}.

\bibitem{ChenM16}
Hubie Chen and Stefan Mengel.
\newblock Counting answers to existential positive queries: {A} complexity
  classification.
\newblock In Tova Milo and Wang{-}Chiew Tan, editors, {\em Proceedings of the
  35th {ACM} {SIGMOD-SIGACT-SIGAI} Symposium on Principles of Database Systems,
  {PODS} 2016, San Francisco, CA, USA, June 26 - July 01, 2016}, pages
  315--326. {ACM}, 2016.
\newblock \href {https://doi.org/10.1145/2902251.2902279}
  {\path{doi:10.1145/2902251.2902279}}.

\bibitem{ChenCVB20}
Zhengdao Chen, Lei Chen, Soledad Villar, and Joan Bruna.
\newblock Can graph neural networks count substructures?
\newblock In Hugo Larochelle, Marc'Aurelio Ranzato, Raia Hadsell,
  Maria{-}Florina Balcan, and Hsuan{-}Tien Lin, editors, {\em Advances in
  Neural Information Processing Systems 33: Annual Conference on Neural
  Information Processing Systems 2020, NeurIPS 2020, December 6-12, 2020,
  virtual}, 2020.
\newblock URL:
  \url{https://proceedings.neurips.cc/paper/2020/hash/75877cb75154206c4e65e76b88a12712-Abstract.html}.

\bibitem{CurticapeanDM17}
Radu Curticapean, Holger Dell, and D{\'{a}}niel Marx.
\newblock Homomorphisms are a good basis for counting small subgraphs.
\newblock In Hamed Hatami, Pierre McKenzie, and Valerie King, editors, {\em
  Proceedings of the 49th Annual {ACM} {SIGACT} Symposium on Theory of
  Computing, {STOC} 2017, Montreal, QC, Canada, June 19-23, 2017}, pages
  210--223. {ACM}, 2017.
\newblock \href {https://doi.org/10.1145/3055399.3055502}
  {\path{doi:10.1145/3055399.3055502}}.

\bibitem{DawarJR21}
Anuj Dawar, Tom{\'{a}}s Jakl, and Luca Reggio.
\newblock {Lov{\'{a}}sz-Type Theorems and Game Comonads}.
\newblock In {\em 36th Annual {ACM/IEEE} Symposium on Logic in Computer
  Science, {LICS} 2021, Rome, Italy, June 29 - July 2, 2021}, pages 1--13.
  {IEEE}, 2021.
\newblock \href {https://doi.org/10.1109/LICS52264.2021.9470609}
  {\path{doi:10.1109/LICS52264.2021.9470609}}.

\bibitem{DellGR18}
Holger Dell, Martin Grohe, and Gaurav Rattan.
\newblock Lov{\'{a}}sz meets weisfeiler and leman.
\newblock In Ioannis Chatzigiannakis, Christos Kaklamanis, D{\'{a}}niel Marx,
  and Donald Sannella, editors, {\em 45th International Colloquium on Automata,
  Languages, and Programming, {ICALP} 2018, July 9-13, 2018, Prague, Czech
  Republic}, volume 107 of {\em LIPIcs}, pages 40:1--40:14. Schloss Dagstuhl -
  Leibniz-Zentrum f{\"{u}}r Informatik, 2018.
\newblock \href {https://doi.org/10.4230/LIPIcs.ICALP.2018.40}
  {\path{doi:10.4230/LIPIcs.ICALP.2018.40}}.

\bibitem{DellRW19}
Holger Dell, Marc Roth, and Philip Wellnitz.
\newblock Counting answers to existential questions.
\newblock In Christel Baier, Ioannis Chatzigiannakis, Paola Flocchini, and
  Stefano Leonardi, editors, {\em 46th International Colloquium on Automata,
  Languages, and Programming, {ICALP} 2019, July 9-12, 2019, Patras, Greece},
  volume 132 of {\em LIPIcs}, pages 113:1--113:15. Schloss Dagstuhl -
  Leibniz-Zentrum f{\"{u}}r Informatik, 2019.
\newblock \href {https://doi.org/10.4230/LIPIcs.ICALP.2019.113}
  {\path{doi:10.4230/LIPIcs.ICALP.2019.113}}.

\bibitem{DellRW19arxiv}
Holger Dell, Marc Roth, and Philip Wellnitz.
\newblock Counting answers to existential questions.
\newblock {\em CoRR}, abs/1902.04960, 2019.
\newblock URL: \url{http://arxiv.org/abs/1902.04960}, \href
  {https://arxiv.org/abs/1902.04960} {\path{arXiv:1902.04960}}.

\bibitem{DourisboureG07}
Yon Dourisboure and Cyril Gavoille.
\newblock Tree-decompositions with bags of small diameter.
\newblock {\em Discret. Math.}, 307(16):2008--2029, 2007.
\newblock \href {https://doi.org/10.1016/j.disc.2005.12.060}
  {\path{doi:10.1016/j.disc.2005.12.060}}.

\bibitem{DurandM15}
Arnaud Durand and Stefan Mengel.
\newblock Structural tractability of counting of solutions to conjunctive
  queries.
\newblock {\em Theory Comput. Syst.}, 57(4):1202--1249, 2015.
\newblock \href {https://doi.org/10.1007/s00224-014-9543-y}
  {\path{doi:10.1007/s00224-014-9543-y}}.

\bibitem{Dvorak10}
Zdenek Dvor{\'{a}}k.
\newblock On recognizing graphs by numbers of homomorphisms.
\newblock {\em J. Graph Theory}, 64(4):330--342, 2010.
\newblock \href {https://doi.org/10.1002/jgt.20461}
  {\path{doi:10.1002/jgt.20461}}.

\bibitem{FlumG06}
J{\"{o}}rg Flum and Martin Grohe.
\newblock {\em Parameterized Complexity Theory}.
\newblock Texts in Theoretical Computer Science. An {EATCS} Series. Springer,
  2006.
\newblock \href {https://doi.org/10.1007/3-540-29953-X}
  {\path{doi:10.1007/3-540-29953-X}}.

\bibitem{Furer01}
Martin F{\"{u}}rer.
\newblock Weisfeiler-lehman refinement requires at least a linear number of
  iterations.
\newblock In Fernando Orejas, Paul~G. Spirakis, and Jan van Leeuwen, editors,
  {\em Automata, Languages and Programming, 28th International Colloquium,
  {ICALP} 2001, Crete, Greece, July 8-12, 2001, Proceedings}, volume 2076 of
  {\em Lecture Notes in Computer Science}, pages 322--333. Springer, 2001.
\newblock \href {https://doi.org/10.1007/3-540-48224-5\_27}
  {\path{doi:10.1007/3-540-48224-5\_27}}.

\bibitem{Grohe21}
Martin Grohe.
\newblock The logic of graph neural networks.
\newblock In {\em 36th Annual {ACM/IEEE} Symposium on Logic in Computer
  Science, {LICS} 2021, Rome, Italy, June 29 - July 2, 2021}, pages 1--17.
  {IEEE}, 2021.
\newblock \href {https://doi.org/10.1109/LICS52264.2021.9470677}
  {\path{doi:10.1109/LICS52264.2021.9470677}}.

\bibitem{GroheKMS21}
Martin Grohe, Kristian Kersting, Martin Mladenov, and Pascal Schweitzer.
\newblock {{Color Refinement and Its Applications}}.
\newblock In {\em {An Introduction to Lifted Probabilistic Inference}}. The MIT
  Press, 08 2021.
\newblock \href
  {https://arxiv.org/abs/https://direct.mit.edu/book/chapter-pdf/2101088/c025000\_9780262365598.pdf}
  {\path{arXiv:https://direct.mit.edu/book/chapter-pdf/2101088/c025000\_9780262365598.pdf}},
  \href {https://doi.org/10.7551/mitpress/10548.003.0023}
  {\path{doi:10.7551/mitpress/10548.003.0023}}.

\bibitem{ImmermanL90}
Neil Immerman and Eric Lander.
\newblock {\em Describing Graphs: A First-Order Approach to Graph
  Canonization}, pages 59--81.
\newblock Springer New York, New York, NY, 1990.
\newblock \href {https://doi.org/10.1007/978-1-4612-4478-3_5}
  {\path{doi:10.1007/978-1-4612-4478-3_5}}.

\bibitem{JinKL18}
Jiashun Jin, Zheng~Tracy Ke, and Shengming Luo.
\newblock Network global testing by counting graphlets.
\newblock In Jennifer~G. Dy and Andreas Krause, editors, {\em Proceedings of
  the 35th International Conference on Machine Learning, {ICML} 2018,
  Stockholmsm{\"{a}}ssan, Stockholm, Sweden, July 10-15, 2018}, volume~80 of
  {\em Proceedings of Machine Learning Research}, pages 2338--2346. {PMLR},
  2018.
\newblock URL: \url{http://proceedings.mlr.press/v80/jin18b.html}.

\bibitem{LanzingerB23}
Matthias Lanzinger and Pablo Barceló.
\newblock On the power of the weisfeiler-leman test for graph motif parameters.
\newblock {\em CoRR}, abs/2309.17053, 2023.
\newblock \href {https://arxiv.org/abs/2309.17053} {\path{arXiv:2309.17053}},
  \href {https://doi.org/10.48550/arXiv.2309.17053}
  {\path{doi:10.48550/arXiv.2309.17053}}.

\bibitem{Lovasz12}
L{\'{a}}szl{\'{o}} Lov{\'{a}}sz.
\newblock {\em Large Networks and Graph Limits}, volume~60 of {\em Colloquium
  Publications}.
\newblock American Mathematical Society, 2012.

\bibitem{MancinskaR20}
Laura Mancinska and David~E. Roberson.
\newblock Quantum isomorphism is equivalent to equality of homomorphism counts
  from planar graphs.
\newblock In Sandy Irani, editor, {\em 61st {IEEE} Annual Symposium on
  Foundations of Computer Science, {FOCS} 2020, Durham, NC, USA, November
  16-19, 2020}, pages 661--672. {IEEE}, 2020.
\newblock \href {https://doi.org/10.1109/FOCS46700.2020.00067}
  {\path{doi:10.1109/FOCS46700.2020.00067}}.

\bibitem{Miloetal02}
R.~Milo, S.~Shen-Orr, S.~Itzkovitz, N.~Kashtan, D.~Chklovskii, and U.~Alon.
\newblock {Network Motifs: Simple Building Blocks of Complex Networks}.
\newblock {\em Science}, 298(5594):824--827, 2002.

\bibitem{MorrisRM20}
Christopher Morris, Gaurav Rattan, and Petra Mutzel.
\newblock Weisfeiler and leman go sparse: Towards scalable higher-order graph
  embeddings.
\newblock In Hugo Larochelle, Marc'Aurelio Ranzato, Raia Hadsell,
  Maria{-}Florina Balcan, and Hsuan{-}Tien Lin, editors, {\em Advances in
  Neural Information Processing Systems 33: Annual Conference on Neural
  Information Processing Systems 2020, NeurIPS 2020, December 6-12, 2020,
  virtual}, 2020.

\bibitem{Morrisetal19}
Christopher Morris, Martin Ritzert, Matthias Fey, William~L. Hamilton, Jan~Eric
  Lenssen, Gaurav Rattan, and Martin Grohe.
\newblock Weisfeiler and leman go neural: Higher-order graph neural networks.
\newblock In {\em The Thirty-Third {AAAI} Conference on Artificial
  Intelligence, {AAAI} 2019, The Thirty-First Innovative Applications of
  Artificial Intelligence Conference, {IAAI} 2019, The Ninth {AAAI} Symposium
  on Educational Advances in Artificial Intelligence, {EAAI} 2019, Honolulu,
  Hawaii, USA, January 27 - February 1, 2019}, pages 4602--4609. {AAAI} Press,
  2019.
\newblock \href {https://doi.org/10.1609/aaai.v33i01.33014602}
  {\path{doi:10.1609/aaai.v33i01.33014602}}.

\bibitem{neuen2023homomorphism}
Daniel Neuen.
\newblock Homomorphism-distinguishing closedness for graphs of bounded
  tree-width.
\newblock {\em arXiv preprint arXiv:2304.07011}, 2023.

\bibitem{PichlerS13}
Reinhard Pichler and Sebastian Skritek.
\newblock Tractable counting of the answers to conjunctive queries.
\newblock {\em J. Comput. Syst. Sci.}, 79(6):984--1001, 2013.
\newblock \href {https://doi.org/10.1016/j.jcss.2013.01.012}
  {\path{doi:10.1016/j.jcss.2013.01.012}}.

\bibitem{roberson2022oddomorphisms}
David~E Roberson.
\newblock Oddomorphisms and homomorphism indistinguishability over graphs of
  bounded degree.
\newblock {\em arXiv preprint arXiv:2206.10321}, 2022.

\bibitem{ScheidtS23}
Benjamin Scheidt and Nicole Schweikardt.
\newblock Counting homomorphisms from hypergraphs of bounded generalised
  hypertree width: {A} logical characterisation.
\newblock In J{\'{e}}r{\^{o}}me Leroux, Sylvain Lombardy, and David Peleg,
  editors, {\em 48th International Symposium on Mathematical Foundations of
  Computer Science, {MFCS} 2023, August 28 to September 1, 2023, Bordeaux,
  France}, volume 272 of {\em LIPIcs}, pages 79:1--79:15. Schloss Dagstuhl -
  Leibniz-Zentrum f{\"{u}}r Informatik, 2023.
\newblock \href {https://doi.org/10.4230/LIPIcs.MFCS.2023.79}
  {\path{doi:10.4230/LIPIcs.MFCS.2023.79}}.

\bibitem{Seppelt23}
Tim Seppelt.
\newblock Logical equivalences, homomorphism indistinguishability, and
  forbidden minors.
\newblock In J{\'{e}}r{\^{o}}me Leroux, Sylvain Lombardy, and David Peleg,
  editors, {\em 48th International Symposium on Mathematical Foundations of
  Computer Science, {MFCS} 2023, August 28 to September 1, 2023, Bordeaux,
  France}, volume 272 of {\em LIPIcs}, pages 82:1--82:15. Schloss Dagstuhl -
  Leibniz-Zentrum f{\"{u}}r Informatik, 2023.
\newblock \href {https://doi.org/10.4230/LIPIcs.MFCS.2023.82}
  {\path{doi:10.4230/LIPIcs.MFCS.2023.82}}.

\bibitem{Tinhofer86}
Gottfried Tinhofer.
\newblock Graph isomorphism and theorems of birkhoff type.
\newblock {\em Computing}, 36(4):285--300, 1986.
\newblock \href {https://doi.org/10.1007/BF02240204}
  {\path{doi:10.1007/BF02240204}}.

\bibitem{Tinhofer91}
Gottfried Tinhofer.
\newblock A note on compact graphs.
\newblock {\em Discret. Appl. Math.}, 30(2-3):253--264, 1991.
\newblock \href {https://doi.org/10.1016/0166-218X(91)90049-3}
  {\path{doi:10.1016/0166-218X(91)90049-3}}.

\bibitem{WL68}
Boris Weisfeiler and Andrei Leman.
\newblock The reduction of a graph to canonical form and the algebra which
  appears therein.
\newblock {\em {NTI}, Series}, 2(9):12--16, 1968.
\newblock English translation by G.\ Ryabov available at
  \url{https://www.iti.zcu.cz/wl2018/pdf/wl_paper_translation.pdf}.

\bibitem{XuHLJ19}
Keyulu Xu, Weihua Hu, Jure Leskovec, and Stefanie Jegelka.
\newblock How powerful are graph neural networks?
\newblock In {\em 7th International Conference on Learning Representations,
  {ICLR} 2019, New Orleans, LA, USA, May 6-9, 2019}. OpenReview.net, 2019.
\newblock URL: \url{https://openreview.net/forum?id=ryGs6iA5Km}.

\end{thebibliography}

\end{document}